\documentclass[sigconf]{acmart}

\AtBeginDocument{%
  \providecommand\BibTeX{{%
    \normalfont B\kern-0.5em{\scshape i\kern-0.25em b}\kern-0.8em\TeX}}}

\copyrightyear{2020}
\acmYear{2020}
\setcopyright{acmcopyright}\acmConference[KDD '20]{Proceedings of the 26th ACM SIGKDD Conference on Knowledge Discovery and Data Mining}{August 23--27, 2020}{Virtual Event, CA, USA}
\acmBooktitle{Proceedings of the 26th ACM SIGKDD Conference on Knowledge Discovery and Data Mining (KDD '20), August 23--27, 2020, Virtual Event, CA, USA}
\acmPrice{15.00}
\acmDOI{10.1145/3394486.3403073}
\acmISBN{978-1-4503-7998-4/20/08}

\usepackage{amsmath,amssymb,xspace,verbatim,paralist,enumitem}
\usepackage{algorithm,algorithmicx}
\usepackage[noend]{algpseudocode}
\usepackage{graphicx}
\usepackage{mathdots,mathtools,mathrsfs}

\usepackage{amsthm}
\usepackage{thmtools}
\usepackage{thm-restate}

\usepackage[normalem]{ulem} 

\usepackage[nameinlink]{cleveref}

\usepackage{stfloats}

\usepackage{subcaption}
\newtheorem{theorem}{Theorem}[section]
\newtheorem{lemma}[theorem]{Lemma}

\theoremstyle{definition}  \newtheorem{definition}[theorem]{Definition}

\newbox\mybox 
\newdimen\myboxwidth    

\newcommand\addpicture[3]{%
\setbox\mybox=\hbox{\includegraphics[scale=#3]{#2}}
\myboxwidth\wd\mybox    
\renewcommand\windowpagestuff{%
\includegraphics[scale=#3]{#2}
\captionof{figure}{A test figure.}}
\parpic[#1]{%
\begin{minipage}{\myboxwidth}
 \windowpagestuff 
\end{minipage} 
} }

\newcommand{\ignore}[1]{}

\newcommand{\mypar}[1]{\medskip\noindent{\sffamily\bfseries #1.}~}
\newcommand{\etal}{{et al.}\xspace}

\newcommand{\eps}{\varepsilon}

\newcommand{\EX}{\hbox{\bf E}}
\newcommand{\Var}{\hbox{\bf Var}}

\newcommand{\EE}{\mathbb{E}}

\newcommand{\ALG}{TETRIS\xspace}
\newcommand{\EdgeALG}{\textsc{EdgeCountEstimator}\xspace}
\newcommand{\seed}{s}
\newcommand{\wL}{r}
\newcommand{\sL}{\ell}
\newcommand{\nt}{t}			
\newcommand{\NT}{T}			
\DeclareMathOperator{\mix}{\ell_{mix}}
\DeclareMathOperator{\emax}{max}

\renewcommand{\ge}{\geqslant}
\renewcommand{\le}{\leqslant}
\renewcommand{\geq}{\geqslant}
\renewcommand{\leq}{\leqslant}

\newcommand{\tO}{\widetilde{O}}

\begin{document}


\title{How to Count Triangles, without Seeing the Whole Graph}

\author{Suman K. Bera}
\email{sbera@ucsc.edu}
\affiliation{%
  \institution{UC Santa Cruz}
  \city{Santa Cruz}
  \state{California}
  \country{USA}
  \postcode{95064}
}
\author{C. Seshadhri}
\email{sesh@ucsc.edu}
\affiliation{%
  \institution{UC Santa Cruz}
  \city{Santa Cruz}
  \state{California}
  \country{USA}
  \postcode{95064}
}



\begin{abstract}
Triangle counting is a fundamental problem in the analysis of large
graphs. There is a rich body of work on this problem, in varying 
streaming and distributed models, yet all these algorithms require
reading the whole input graph. In many scenarios,
we do not have access to the whole graph, and can only 
sample a small portion of the graph (typically through crawling).
In such a setting, how can we accurately estimate the triangle count
of the graph?

We formally study triangle counting in the {\em random walk} access model
introduced by Dasgupta et al (WWW '14) and Chierichetti et al (WWW '16).
We have access to an arbitrary seed vertex of the graph, and can only perform
random walks. This model is restrictive in access and captures the challenges
of collecting real-world graphs. Even sampling a uniform random vertex is a hard task in this model.

Despite these challenges, we design a provable and practical algorithm, TETRIS,
for triangle counting in this model. TETRIS is the first provably sublinear algorithm
(for most natural parameter settings) that approximates the triangle count
in the random walk model, for graphs with low mixing time. 
Our result builds on recent advances in the theory
of sublinear algorithms. The final sample built by TETRIS is a careful mix
of random walks and degree-biased sampling of neighborhoods. 
Empirically, TETRIS accurately counts triangles on 
a variety of large graphs, getting estimates within 5\% relative error by looking
at 3\% of the number of edges. 



\end{abstract}



\maketitle
\section{Introduction}
\label{sec:intro}

Triangle counting is a fundamental problem in the domain of  
network science. The triangle count (and variants thereof)
appears in many classic parameters in social network analysis
such as the {\em clustering coefficient}~\cite{Ne03}, 
{\em transitivity ratio}~\cite{wasserman1994}, {\em local
clustering coefficients}~\cite{watts1998collective}.
Some example applications of this problem 
are motifs discovery in 
complex biological networks~\cite{milo2002network}, 
modeling large graphs~\cite{SeKoPi11,PfFo+12,PfMo+14},
indexing graph databases~\cite{khan2011neighborhood}, 
and spam and fraud detection cyber security~\cite{BecchettiBCG08}.
Refer to the tutorial~\cite{SeTi19} for more applications.

Given full access to the input graph $G$, we can exactly
count the number of triangles in $O(m^{3/2})$
time~\cite{itai1978finding}, where $m$ is the number of
edges in $G$. Exploiting degree based ordering,
the runtime can be improved to  $O(m\alpha)$~\cite{Chiba1985},
where $\alpha$ is the maximum core number (degeneracy) of $G$. 
In various streaming and distributed models, 
the triangle counting problem has a rich theory~\cite{BarYossefKS02, Jowhari2005, Buriol2006,cohen2009graph,Suri2011,kolda2014counting, McGregor2016, bera2017towards,BeraDegeneracy},
and widely used practical algorithm~\cite{BecchettiBCG08,jha2013space,PavanTTW13,tangwongsan2013parallel,ahmed2014graph,chen2016general,stefani2017triest,turkoglu2017edge,Turk2019}. 

Yet the majority of these algorithms, at some point, read the entire graph. (The only exceptions
are the MCMC based algorithms~\cite{rahman2014sampling,chen2016general}, but they require global
parameters for appropriate normalization. We explain in detail later.) In many practical scenarios, the entire graph is not known. 
Even basic graph parameters such as the total number of vertices and the number of edges are unknown.
Typically, a sample of the graph is obtained by crawling the graph,
from which properties of the true graph must be inferred. In common network analysis settings, practitioners crawl some portion of the (say)
coauthor, web, Facebook, Twitter, etc. network. They perform experiments on this sample, in the hope of inferring
the ``true" properties of the original graph.
This sampled
graph may be an order of magnitude smaller than the true graph.



There is a rich literature on graph sampling, but arguably,
Dasgupta~\etal~\cite{DaKu14} gave the first formalization of 
this sampling through 
the {\em random walk access} model. In this model, we have
access to some arbitrary seed vertex. We can discover portions of the
graph by performing random walks/crawls starting at this vertex. At
any vertex, we can retrieve two basic pieces of information --- its degree
and a uniform random neighbor. Rudimentary graph mining tasks such as
sampling a uniform random node is non-trivial in this model. 
Dasgupta~\etal~\cite{DaKu14} showed how to find the average degree
of the graph in this model. Chierichetti~\etal~\cite{ChDa+16} gave an elegant algorithm for sampling a uniform random vertex. Can we estimate 
more involved graph properties {\em efficiently} in this model?
This leads to the main research question behind this work.

{\em How can we accurately estimate the triangle count of the graph
by observing only a tiny fraction of it in the random walk access
model?}

There is a rich literature of sampling-based
triangle counting algorithms~\cite{schank2005finding,Buriol2006,tsourakakis2009doulion,jha2013space,PavanTTW13,seshadhri2014wedge,ahmed2014graph,turkoglu2017edge,Turk2019}. However,
all of these algorithms heavily relies on
uniform random edge samples or uniform random vertex samples.
Such samples are computationally expensive to
generate in the random walk access model. 
To make matters worse, we do not know the number 
of vertices or edges in the graph. Most sampling algorithms require 
these quantities to compute their final estimate.

\subsection{Problem description}
In this paper, we study the triangle estimation problem.
Given limited access to an input graph $G=(V,E)$, our goal is to design an $(\eps,\delta)$-estimator
for the triangle count.
\begin{definition}
\label{def:tri_cnt}
Let $\eps,\delta \in [0,1]$ be two parameters and $\NT$ denote the triangle counting.
A randomized algorithm is an $(\eps,\delta)$-estimator for the triangle counting problem if:
the algorithm outputs estimate $\overline{\NT}$
such that with probability (over the randomness of the algorithm) at least $1-\delta$, $(1-\eps) \NT \leq \overline{\NT} \leq (1+\eps) \NT$.
(We stress that there is no stochastic assumption on the input itself.)
\end{definition}
{\bf The {\em random walk access} model:} 
In designing an algorithm, we aim to minimize access to 
the input graph. To formalize, we required a query model and follow
models given in  Dasgupta~\etal~\cite{DaKu14} and Chierichetti~\etal~\cite{ChDa+16}.

The algorithm initially has access to a single, \emph{arbitrary} seed vertex $s$.
The algorithm can make three types of queries:
\begin{asparaitem}
    \item {\em Random Neighbor Query}: Given a vertex $v\in V$, acquire
    a uniform random neighbor of $v$.
    \item {\em Degree Query}: Given a vertex $v\in V$, acquire the degree of $v$.
    \item {\em Edge Query}: Given two vertices $u,v\in V$, acquire whether 
    the edge $\{u,v\}$ is present in the graph.
\end{asparaitem}
Starting from $s$, the algorithm makes queries to discover
more vertices, makes queries from these newly discovered vertices to
see more of the graph, so on and so forth.
An algorithm in this model does not have free access to the set of vertices.
It can only make {\em degree queries} and {\em edge queries}
on the vertices that have been observed during the random walk.
We emphasize that the random walk does not start from an uniform random 
vertex; the seed vertex is arbitrary. In fact, generating a uniform vertex in this 
model is a non-trivial task and explicitly studied by Chierichetti~\etal~\cite{ChDa+16}
We note a technical difference between the above model and those studied in~\cite{DaKu14, ChDa+16}.
Some of these results assume all neighbors can be obtained with a single query,
while we assume a query gives a random neighbor. As a first step towards sublinear
algorithms, we feel that the latter is more appropriate
to understand how much of graph needs to be sampled to estimate the triangle counting.
But it would be equally interesting to study different query models.
    

We do not assume prior 
knowledge on the number of vertices or edges in the graph.
This is consistent with many practical scenarios, such as API based online
networks, where estimating $|E|$ itself is a challenging task.
The model as defined is extremely restrictive, making
it challenging to design 
provably accurate algorithms in this model.

{\bf On uniform vertex/edge samples:} There
is a large body of theoretical and practical work efficient 
on efficient
triangle counting assuming uniform random vertex or edge samples (the latter is often
simulated by streaming algorithms)~\cite{BarYossefKS02,Jowhari2005,schank2005finding,Buriol2006,tsourakakis2009doulion, kolountzakis2012efficient,wu2016counting,McGregor2016, bera2017towards,BeraDegeneracy}.
However, when uniform random vertex or edge samples are not
available, as is the case in our model,
the existing literature is surprisingly quiet.

{\bf Complexity measure:} While designing an $(\eps,\delta)$-estimator
in the random walk model, our goal is to minimize the number of
\emph{queries} made. We do not deal with running time (though our final algorithm
is also time-efficient). For empirical results, we express queries as a fraction
of $m$, the total number of edges. In mathematical statements,
we sometimes use the $\tO$ notation to hide
dependencies on $\eps,\delta,\log n$.
%
\begin{figure}[!ht]
  \centering
  \includegraphics[width=0.7\columnwidth]{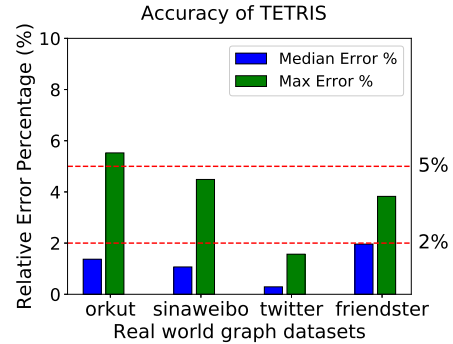}
  \caption{Accuracy of \ALG on real datasets: soc-orkut (3M vertices, 213M edges), soc-siaweibo (59M vertices, 523M edges), soc-twitter(42M vertices, 2.4B edges), soc-friendster(66M vertices, 3.6B edges). We run \ALG with exact same set of parameters for all the datasets and observes only 3\% of the edges. We repeat the experiments for 100 times. Remarkably for 100 independent runs, \ALG achieves worst case relative error of about 5\% (the green bar in the plot) and median relative error of 2\% (the blue bar in the plot).}
  \label{fig:intro-fig}
\end{figure}

\subsection{Our contributions}
In this work, we present a novel algorithm, Triangle Estimation Through Random
Incidence Sampling, or \ALG, that solves the triangle counting problem 
efficiently in the random walk access model. \ALG provably outputs
an $(\eps,\delta)$-estimator for the triangle counting problem. Under common assumptions
on the input graph $G$, \ALG is provably sublinear.
In practice, \ALG is highly accurate and robust.
Applying to real datasets, we demonstrate that,
it only needs to observe a tiny fraction of the graph
to make an accurate estimation of the triangle count (see
~\Cref{fig:intro-fig}).

\mypar {First provably sublinear triangle counting algorithm in the {\em random walk access} model} 
Our algorithm \ALG builds on recent theoretical results in 
sublinear algorithms for clique counting~\cite{ELRS15,eden2020faster}. 
Our central idea is to sample edges proportional to a somewhat non-intuitive quantity:
the degree of the lower degree endpoint of the edge.
In general, such weighted sampling is impossible to achieve in the {\em random walk access} model. 
However, borrowing techniques from Eden~\etal~{\cite{eden2020faster}},
we approximately simulate such a sample in this
model. By a careful analysis of various components, we prove that
\ALG  is an $(\eps,\delta)$-estimator.

%
\mypar{Accurate and robust empirical behavior} We run \ALG on 
a collection of massive datasets, each with more than 100 million edges. For all instances, \ALG 
achieves less than 3\% median relative error while observing less than
3\% of the graph. Results are shown in~\Cref{fig:intro-fig}. 
Even over a hundred runs (on each dataset), \ALG has a maximum error
of at most 5\%. We also empirically demonstrate the robustness of our algorithm 
against various choices for the seed vertex.  
   
\mypar{Comparison with existing sampling based methods} While the vast majority
of triangle counting algorithm read the whole graph, the MCMC algorithms
of Rahman~\etal~\cite{rahman2014sampling} and Chen~\etal~\cite{chen2016general}
can be directly adapted to the random walk access model. We also note
that the sparsification algorithm of Tsourakakis~\etal~\cite{tsourakakis2009doulion}
and the sampling method of Wu~\etal~\cite{wu2016counting} can be implemented
(with some changes) in the random walk model. 
We perform head to head comparisons of \ALG
with these algorithms. We observe that \ALG is the only algorithm that consistently
has a low error in all instances. The Subgraph Random walk of Chen~\etal~\cite{chen2016general}
has a reasonable performance across all instances, though its error is typically double
of \ALG. All other algorithms have an extremely high error on some input.
These findings are consistent with our theoretical analysis of \ALG, which proves
sublinear query complexity (under certain conditions on graph parameters).

%
%
\subsection{Main Theorem}
Our main theoretical contribution is to prove that
Triangle Estimation Through Random
Incidence Sampling, or \ALG is an $(\eps,\delta)$-estimator for the triangle
counting problem. 
Let $G$ be the input graph with $m$ edges and maximum
core number (degeneracy) $\alpha$.
\footnote{Degeneracy of a graph $G$ is the smallest integer 
$k$, such that every subgraph of $G$ has a vertex of degree
at most $k$. It is often called the maximum core number.
Informally, 
degeneracy is a measure of the sparsity of a graph.}
We use $\nt_e$ to denote the number of triangles incident on an
edge $e$. Define the maximum value of $\nt_e$ as $\nt_{\emax} := \max_{e\in E} \nt_e $. Assume the total number of triangles is $\NT$.
We denote the mixing time of the input graph as $\mix$.

\begin{theorem}
\label{thm:tri_count}
Let $\eps,\delta \in[0,1]$ be some parameters and $G$ be an arbitrary input graph. Then \ALG produces
an $(\eps,\delta)$-estimate for the triangle count in the random walk access model.
Moreover, the number of
queries made by \ALG is at most
\begin{align*}
O\left(\log \left(\frac{1}{\delta}\right) \frac{\log n}{\eps^2}
\left( \frac{m \mix \nt_{\emax}}{\NT} + \frac{m \alpha} {\NT} + \sqrt{m} \mix\right) \right) \,.
\end{align*}
The runtime and the space requirement of \ALG is bounded by the
same quantity as well.
\end{theorem}
The quantities $\alpha$ and $\mix$ are considered small in social networks. Moreover,
$\nt_{\emax}$ is much smaller than the total triangle count $\NT$. Thus, the bound above
is strictly sublinear for social networks. Moreover, it is known that lower bound
of $\mix$ is necessary even for the simpler problem of sampling vertices~\cite{chierichetti2018complexity}.
\subsection{Main Ideas and Challenges}

Sampling based triangle counting algorithms typically work as follows.
One associates some quantity (like the triangle count of an edge) to each edge,
and estimates the sum of this quantity by randomly sampling edges. If uniform
random edges are present, then one can use standard scaling to estimate 
the sum (\cite{tsourakakis2009doulion,kolountzakis2012efficient,wu2016counting}).
Other methods set up a highly non-uniform distribution, typically biasing more heavily
to higher degree vertices, in the hope of catching triangles with fewer samples
(the classic method being wedge sampling~\cite{schank2005finding,seshadhri2014wedge,rahman2014sampling,turkoglu2017edge}).
This is more efficient in terms of samples, but requires more complex sampling and needs non-trivial normalization
factors, such as total wedge counts. All of these techniques face problems in the random
walk access model. Even if the mixing time $\mix$ is low, one needs $k\mix$ queries to
get $k$ truly independent uniform edges. And it could potentially require even more samples
to get non-uniform wedge based distributions. Indeed, our experiments 
show that this overhead leads to inefficient algorithms in the random walk access model.

To get around this problem, we use the key idea of ordering vertices by degree,
a common method in triangle counting. Inspired by a very recent sublinear clique counting algorithms
of Eden~\etal~\cite{eden2020faster}, we wish to sample edges proportional to the degree
of the smaller degree endpoint. Such ideas have been used to speed up wedge sampling~\cite{turkoglu2017edge,Turk2019}.
Such a sampling biases away from extremely high degree vertices, which
is beneficial in the random walk model.

But how to sample according to this (strange) distribution? Our algorithm
is rather naive: simply take a long random walk, and just perform this sampling
\emph{among} the edges of the random walk. Rather surprisingly, this method
provably works. The proof requires a number of ideas introduced in~\cite{eden2020faster}.
The latter result requires access to uniform random vertices, but we are able to adapt
the proof strategy to the random walk model.

The final algorithm is a direct implementation of the above ideas, though there
is some technicality about the unique assignment of triangles to edges. The proof,
on the other hand, is quite non-trivial and has many moving parts. We need to show
that sampling according to the biased distribution among the edges
of the random walk has similar statistical properties to sampling from all edges.
This requires careful use of concentration bounds to show that various statistics
of the entire set of edges are approximately preserved in the set of random walk edges.
Remarkably, the implementation follows the theoretical algorithm exactly. We require no
extra heuristics to make the algorithm practical.

We mention some benefits of our approach, in the context of practical sublinear algorithms.
We do not need non-trivial scaling factors, like total wedge counts (required for any wedge sampling
approach). Also, a number of practical MCMC methods perform random walks on ``higher order"
Markov Chains where states are small subgraphs of $G$. Any theoretical analysis requires
mixing time bounds on these Markov Chains, and it is not clear how to relate these bounds
to the properties of $G$. Our theoretical analysis relies only on the mixing time of $G$,
and thus leads to a cleaner main theorem. Moreover, a linear dependence of the mixing time
is likely necessary, as shown by recent lower bounds~\cite{chierichetti2018complexity}.

\section{Related Work}

There is an immense body of literature on the triangle counting problem,
and we only focus on work that is directly relevant.
For a detailed list of citations, we refer the reader to the tutorial~\cite{SeTi19}.

There are a number of exact triangle counting algorithms,
often based on the classic work of Chiba-Nishizeki~\cite{Chiba1985}.
Many of these algorithms have been improved by using clever heuristics, parallelism,
or implementations that reduce disk I/O~\cite{schank2005finding,hu2014efficient,green2014fast,azad2015parallel,kim2016dualsim}.

Approximate triangle counting algorithms are significantly
faster and more widely studied in both theory and practice. 
Many techniques have been explored to solve this problem 
efficiently in the streaming 
model, the map-reduce model, and the general static model (or RAM model).
There are plethora of practical
streaming algorithms~\cite{BecchettiBCG08,jha2013space,PavanTTW13,tangwongsan2013parallel,ahmed2014graph,jha2015counting,stefani2017triest},
MapReduce algorithms~\cite{cohen2009graph,Suri2011,kolda2014counting}, and other distributed models~\cite{kolountzakis2012efficient,arifuzzaman2013patric} as well. 
Even in the static settings,
sampling based approaches have been proven quite efficient
~\cite{tsourakakis2009doulion,seshadhri2014wedge,etemadi2016efficient,turkoglu2017edge,Turk2019} for this problem.

\emph{All} of these algorithms read the entire input, with the notable exception of MCMC based algorithms~\cite{rahman2014sampling,chen2016general}. 
These algorithms perform a random walk in a ``higher" order graph, that can be locally constructed by looking at the neighborhoods of a vertex.
We can implement these algorithms in the random walk access model, and indeed, consider these to be the state-of-the-art triangle counting
algorithms for this model. We note that these results are not provably sublinear (nor do they claim to be so). Nonetheless, we find they perform
quite well in terms of making few queries to get accurate estimates for the triangle count.
Notably, the Subgraph Random Walk algorithm of~\cite{chen2016general} is the only other algorithm
that gets reasonable performance in all instances, and the VertexMCMC algorithm of~\cite{rahman2014sampling} is the only algorithm that ever outperforms \ALG (though in other instances, it does not seem
to converge).

We note that the Doulion algorithm of Tsourakakis~\etal~\cite{tsourakakis2009doulion}
and the sampling method of Wu~\etal~\cite{wu2016counting} can also be implemented in the random
walk access model. Essentially, we can replace their set of uniform random edges with the set
of edges produced by a random walk (with the hope that after the mixing time, the set ``behaves"
uniform). We observe that these methods do not perform too well.

From a purely theoretical standpoint, there have been recent sublinear algorithms for triangle counting
by Eden~\etal~\cite{ELRS15, eden2020faster}. These algorithms require access to uniform random samples
and cannot be implemented directly in the random walk model.
Nonetheless, the techniques from these results are highly relevant for \ALG.
In particular, we borrow the two-phase sampling idea
from~\cite{ELRS15,ERS17} to simulate the generation of edge samples according
to edge based degree distribution. 

The {\em random walk access model} is formalized by
Dasgupta~\etal~\cite{DaKu14} in the context of the average
degree estimation problem. Prior to that, there have 
been several works on estimating basic statistics of a massive network
through {\em crawling}. Katzir~\etal\cite{katzir2011estimating},
Hardiman~\etal~\cite{hardiman2009calculating}, and 
Hardiman and Katzir~\cite{katzir2015estimating} used collisions
to estimate graph sizes and clustering coefficients of individual vertices.
Cooper~\etal~\cite{cooper2014estimating} used random walks for estimating
various network parameters, however, their sample complexity is 
still is at least that of collision based approaches.
Chierichetti~\etal~\cite{ChDa+16} studied the problem of
uniformly sampling a node in the {\em random walk access} model. 
Many triangle counting estimators are built on uniform random node 
samples. However, using the uniform node sampler of  
Chierichetti~\etal~\cite{ChDa+16} leads to an expensive triangle 
estimator, because of the overhead of generating vertex samples.

There are quite a few sampling methods based on random crawling:
forest-fire~\cite{LF06}, snowball sampling~\cite{maiya2011benefits}, 
and expansion sampling~\cite{LF06}. However, they do not lead
to unbiased estimators for the triangle counting problem.
It is not clear whether such sampling methods can generate
provably accurate estimator for this problem.
For a more detailed survey of various sampling
methods for graph parameter estimation, we refer 
to the works of Leskovec and  Faloutsos~\cite{LF06}, 
Maiya and Berger-Wolf~\cite{maiya2011benefits}, 
and Ahmed, Neville, and Kompella~\cite{ahmed2014network}. 

\section{Preliminaries}
\label{sec:prelim}
In this paper, graphs are undirected and unweighted.
We denote the input graph by $G=(V,E)$, and put $|V|=n$ and
$|E|=m$. We denote the number of triangles in the input graph
by $\NT$. For an integer $k$, we denote the set $\{1,2,,\ldots k\}$
by $[k]$. All the logarithms are base $2$.

For a vertex $v\in V$, we denote its neighborhood as $N(v) := \{u:\{u,v\}\in E\}$, and degree as $d(v) := |N(v)|$. For an
edge $e =\{u,v\}\in E$, we define degree of $e$ as the minimum
degree of its endpoint: $d(e):= \min\{d(u),d(v)\}$. Similarly,
we define the neighborhood of $e$: $N(e) = N(u)$ if $d(u) < d(v)$,
$N(e) = N(v)$ otherwise. 

We consider a degree based ordering of the vertices $\prec_{\deg}$:
for two vertices $u,v\in V$, $u \prec_{\deg} v$ iff $d(u) < d(v)$
or $d(u)=d(v)$ and $u$ precedes $v$ according some fixed ordering, say
lexicographic ordering. Every triangle is uniquely assigned to an edge
as follows. For a triangle $\{v_1, v_2, v_3\}$ such that $v_1 \prec_{\deg} v_2 \prec_{\deg} v_3$,
we assign it to the edge $(v_1, v_2)$.
We denote
the number of triangles associated with an edge $e$ by $\nt_e$.
Clearly, $\sum_{e\in E} \nt_e = \NT$. We denote by $\nt_{\max}$
the maximum number of triangles associated with any edge:
$\nt_{\max} = \max_{e\in E} \nt_e$.

We extend the notion of $d_e$ and $t_e$ to a collection of
edges $R$ naturally: $d_R = \sum_{e\in R} d_e$ and 
$\nt_R = \sum_{e\in R} t_e$. Note that if $R$ is a multi-set,
then we treat every occurrences of an edge $e$ as a distinct 
member of $R$, and the quantities $d_R$ and $\nt_R$ reflect these.

We denote the degeneracy (or the maximum core number) of the input graph by $\alpha$.
Degeneracy, or the maximum core number, is the smallest integer
$k$, such that for every subgraph in the input graph, there is a
vertex of degree at most $k$.
Chiba and Nishizeki proved the following connection between
$\alpha$ and $d_E$.
\begin{lemma}[Chiba and Nishizeki~\cite{Chiba1985}]
\label{lem:chiba}
$d_E = \sum_{e\in E}d_e = O(m \alpha$).
\end{lemma}

We revisit a few basic notions about the random walk on a connected, undirected, 
non-bipartite graph. For every graph, $\pi(v) = \frac{d(v)}{2m}$
is the stationary distribution. We denote
the mixing time of the input graph by $\mix$. 

We use the following concentration bounds for analyzing our algorithms.
In general, we use the shorthand $A \in (1\pm \eps)B$ for $A \in [(1-\eps)B, (1+\eps)B]$.
\begin{theorem} \label{thm:conc}
\begin{asparaitem}
\item [{\bf Chernoff Bound}:] 
  Let $X_1,X_2,\ldots,X_r$ be mutually independent indicator random variables with expectation $\mu$. Then, for every $\eps$ with $0< \eps <1$,
$\Pr\left[ \sum_{i=1}^{r} X_i/r \notin (1\pm\eps) \mu\right] \le 2\exp\left( - \eps^2 r \mu  / 3 \right)$.
\item [{\bf Chebyshev Inequality}:] 
Let $X$ be a random variable with expectation $\mu$
and variance $\Var[X]$. Then, for every $\eps>0$,
$\Pr\left[ X \notin (1\pm \eps)\mu\right] \leq \frac{\Var[X]}{ \eps^2 \mu^2}$.
\end{asparaitem}
\end{theorem}

%

\section{The Main Result and \ALG}
We begin with the description of our triangle counting algorithm,
\ALG(~\Cref{alg:triangle_estimate}). 

It takes
three input parameters: the length of the random walk $\wL$,
the number of subsamples (explained below) $\sL$, and an estimate $\widehat{\mix}$ of the mixing time.
\ALG starts with an arbitrary vertex of the graph provided by the model.
Then it performs a random walk of length $\wL$ and collects the
edges in a ordered multi-set $R$. For each edge $e =\{u,v\}\in R$,
it computes the degree $d_e$. Recall the definition
of degree of an edge: $d_e= \min \{d_u,d_v\}$. 
Then, \ALG samples $\sL$ edges from $R$, where edge $e$
is sampled with probability proportional to $d_e$.

For an edge $e$ sampled in the above step, \ALG
samples a uniform random neighbor $w$ from $N(e)$; recall $N(e)$ denote the neighbors of the lower degree end point. 
Finally, using edge queries, \ALG checks whether a
triangle is formed by $\{e,w\}$. If it forms a triangle, 
\ALG checks
if the triangle is uniquely assigned to the edge $e$ by querying the degrees of the constituent vertices (see~\Cref{sec:prelim} for the assignment rule).

To compute the final estimate for triangles, \ALG requires 
an estimate for the number of edges. 
To accomplish this task, we design a collision based edge estimator
in~\Cref{alg:edge_estimate}, based on a result of Ron and Tsur~\cite{RT16}.

\begin{algorithm}[!ht]
\caption{\ALG --- Triangle Counting Estimator}
\label{alg:triangle_estimate}
    \begin{algorithmic}[1]
    \Procedure{\ALG}{integer $\wL$, integer $\sL$, integer $\widehat{\mix}$}
    \label{proc:EstTri}
    \State Let $\seed$ be some arbitrary vertex provided by the model.
    \State Let $R$ be the multiset of edges on a $\wL$-length random walk from $\seed$.
    \For {$i=1$ to $\sL$} \label{line:loop}
    \State Sample an edge $e\in R$ independently with prob. $d_e/d_R$.
    \State Query a uniform random neighbor $w$ from $N(e)$.
    \State Using edge query, check if $\{e,w\}$ forms a triangle.
    \State If $\{e,w\}$ forms a triangle $\tau$: query for degrees of all vertices
    in $\tau$ and determine if $\tau$ is associated to $e$.
    \State If $\tau$ is associated to $e$, set $Y_i=1$; else set $Y_i=0$.
    \EndFor
    \State Set $Y = \frac{1}{\sL} \sum_{i=1}^{\ell} Y_i$. \label{line:Y}
    \State Let $\overline{m}$ = \textsc{EdgeCountEstimator}$~(R,\widehat{\mix})$
    \State Set $X = \frac{\overline{m}}{\wL} \cdot d_R \cdot Y$.
    \State return $X$.
    \EndProcedure
    \end{algorithmic}
\end{algorithm}

\begin{algorithm}[!ht]
\caption{Edge Count Estimator}
\label{alg:edge_estimate}
    \begin{algorithmic}[1]
    \Procedure{\textsc{EdgeCountEstimator}}{edge set $R$, integer $\widehat{\mix}$}
        \label{proc:EstEdge}
        \For {$i=1$ to $\widehat{\mix}$}
        \State Let $R_i = (e_i,e_{i+\widehat{\mix}},e_{i+2\widehat{\mix}},\ldots)$. \label{line:R_i)}
        \State Set $c_i=$number of pairwise collision in $R_i$. \label{line:R_i}
        \State Set $Y_i= \binom{|R_i|}{2} / c_i$.
        \EndFor
        \State Set $Y=\frac{1}{\widehat{\mix}} \sum_{i} Y_i$.
        \State return $Y$.
        \EndProcedure
    \end{algorithmic}
\end{algorithm}

We state the guarantees of \textsc{EdgeCountEstimator}. It is a direct consequence
of results of Ron and Tsur~\cite{RT16}, and we defer the formal proof to the~\Cref{appendix:proof}.

\begin{theorem}
\label{thm:edge_estimator}
Let $\eps >0$ be some constant, $\widehat{\mix} \geq \mix$, and
$|R| \geq \frac{\log n}{\eps^2}\cdot \mix \cdot \sqrt{m}$. Then, \textsc{EdgeCountEstimator} outputs $\overline{m} \in (1\pm \eps)m$ with probability at least $1-o(1/n)$.
\end{theorem}

\subsection{Theoretical Analysis of \ALG}
We provide a theoretical analysis of \ALG and prove ~\Cref{thm:tri_count}. 
We first show that if the collection of edges $R$ exhibits some ``nice" properties, then
we have an accurate estimator. Then we prove that the collection of edges 
$R$ produced by the random walk has these desired properties.
Our goal is to prove that the output of
\ALG, $X$, is an $(\eps,\delta)$-estimator for the triangle counts.
For ease of exposition, in setting parameter values of 
$\wL$  and $\sL$, we hide the dependency
on the error probability $\delta$. A precise calculation
would set the dependency to be $\log (1/\delta)$, as standard
in the literature.

We first show that the random variable $Y$ (~\cref{line:Y})
roughly captures the ratio $\NT/d_E$. To show this,
we first fix an arbitrary collection of edges $R$,
and simulate \ALG on it. We show that, in expectation,
$Y$ is going to $t_R/d_R$. For the sake of clear presentation,
we denote the value of the random variable $Y$, when run
with the edge collection $R$, to be $Y_R$. Note that, $Y_R$
is a random variable nevertheless; the source of the
randomness lies in the $\sL$ many random experiments that
\ALG does in each iteration of the for loop at ~\cref{line:loop} of ~\Cref{alg:triangle_estimate}
\begin{lemma}
\label{lem:triangle_est}
Let $R$ be a fixed collection of edges, and $Y_R$ denote the value of the random variable $Y$ on the fixed set $R$ (~\cref{line:Y} of ~\Cref{alg:triangle_estimate}).  Then,
\begin{enumerate}
    \item $\EX [Y_R] = \frac{\nt_R}{d_R}$,
    \item $\Pr \left[~|Y_R - \EX[Y_R]| \geq \eps \EX [Y_R]~ \right] 
    \leq \exp \left( - \sL \cdot \frac{\eps^2}{3} \cdot \frac{\nt_R}{d_R}  \right)$.
\end{enumerate}
\end{lemma}
\begin{proof}
We first prove the expectation statement. 
Let $e_i$ be the random variable that denotes the edge sampled
in the $i$-th iteration of the for loop (~\cref{line:loop}).
Consider the random variable $Y_i$.
We have
\begin{align*}
    \Pr[Y_i=1] &= \sum_{e\in R} \Pr[e_i = e]\Pr [Y_i =1 | e_i = e] \,,  \\
    &= \sum_{e\in R} \frac{d_e}{d_R} \Pr [Y_i =1 | e_i = e] = \sum_{e \in R} \frac{\nt_e}{d_R} = \frac{\nt_R}{d_R} \,.
\end{align*}
Where the second last equality follows as for a fixed edge $e$, the probability that $Y_i=1$ is exactly $t_e/d_e$.
Since $Y= (1/\sL) \sum_{i=1}^{\sL} Y_i $,
by linearity of expectation, we have the item (1) of the lemma.

For the second item, we apply the Chernoff bound in~\Cref{thm:conc}.
\end{proof}

Clearly, for any arbitrary set $R$, we may not have the desired
behavior of $Y$ that it concentrates around $\NT/d_E$. 
To this end, we first define the desired properties of $R$ and
then show \ALG produces such an edge collection $R$ with high 
probability.
\begin{definition}[A {\em good} collection of edges]
We call an edge collection $R$ {\em good}, if it satisfies the 
following two properties:
\begin{align}
    &\frac{\nt_{R}}{d_R} \ge (1-\eps) \cdot \frac{\eps}{\log n} \cdot \frac{\NT}{d_E} \label{eqn:prop_1} \\
    & \nt_R \in \left[ (1-\eps) |R| \cdot \frac{\NT}{m} ~~,~~  (1+\eps) |R| \cdot \frac{\NT}{m} \right] \label{eqn:prop_2}
\end{align}
\end{definition}
For now we assume the edge collection $R$
produced by \ALG is {\em good}. Observe that,
under this assumption, the expected value of 
$Y_R$ is  $\nt_R/d_R \geq \tO(\NT/d_E)$. In the next lemma, we 
show that for the setting of $\sL= \tO(d_E/\NT)$, $Y_R$
concentrates tightly around its mean.

\begin{lemma}
\label{lem:goodY}
Let $0 < \eps < 1/2$ and $c>6$ be some constants, and $\sL = \frac{c \log^2 n}{\eps^3}\cdot \frac{d_E}{\NT}$. 
Conditioned on $R$ being {\em good}, 
with probability at least $1- o(1/n)$, 
$|Y_R - \EX[Y_R]| \leq \eps \EX[Y_R] $.
\end{lemma}
\begin{proof}
Since $R$ is {\em good}, by the first property (~\cref{eqn:prop_1}),
we have 
$\frac{t_R}{d_R} \geq (1-\eps) \cdot \frac{\eps}{\log n} \cdot \frac{\NT}{d_E} $. 
Then, by item (2) in~\Cref{lem:triangle_est}, we have
\begin{align*}
    &\Pr \left[~|Y_R - \EX[Y_R]| \geq \eps \EX [Y_R]~ \right] \\
    & \leq \exp \left( - \frac{c \log^2 n}{\eps^3}\cdot \frac{d_E}{t} 
    \cdot\frac{\eps^2}{3} \cdot (1-\eps) \cdot \frac{\eps}{\log n} \cdot \frac{t}{d_E}  \right) = \frac{1}{n^ {\frac{c(1-\eps)}{3}}}\,,
\end{align*}
where the second inequality follows by plugging in the value of $\sL$.
The lemma follows by the constraints on the values of $c$ and $\eps$.
\end{proof}

We now show that, conditioned on $R$ being {\em good},
the final estimate $X$ is accurate. 

\begin{lemma}
\label{lem:tri}
Condition on the event that $R$ is good. Then,
with probability at least $1-o(1/n)$, 
$X_R \in (1\pm4\eps)\NT$.
\end{lemma}
\begin{proof} 
Recall that $\EE[Y_R] = \nt_R /d_R$.
Conditioned on the  event of $R$ being good,
with probability at least $1-o(1/n)$, $Y_R$ is closely
concentrated around its expected value (~\Cref{lem:goodY}).
Hence, we have
$Y_R \in (1\pm\eps) \nt_R/d_R$.
Now consider the final estimate $X$. For a fixed set $R$, denote
the value of the random variable $X$ to be $X_R$. Note that,
$X_R$ itself is a random variable where the source of the
randomness is same as that of $Y_R$. More importantly,
$X_R$ is a just a scaling of the random variable $Y_R$:
$X_R = (\overline{m}/\wL)\cdot d_R \cdot Y_R$.
Then, with high probability
    $X_R \in (1\pm\eps)\frac{\overline{m}}{\wL} \cdot d_R \cdot \frac{\nt_R}{d_R}$.

Simplifying we get, with high probability,
$X_R \in (1\pm\eps)\frac{\overline{m}}{\wL} \cdot  {\nt_R}$.
Since $R$ is good, by using the second property (~\cref{eqn:prop_2}), 
and setting $|R|=r$, $X_R \in (1\pm2\eps)\frac{\overline{m}}{m} \cdot  {\NT}$.
By~\Cref{thm:edge_estimator}, with probability at least $1-o(1/n)$, $\overline{m} \in [(1-\eps)m~~,~~(1+\eps)m]$. Hence,
with probability at least $1-o(1/n)$,
$X_R \in \left[ (1-4\eps) \NT, (1+4\eps)\NT  \right]$.
\end{proof}

We now show that with probability at least $1-1/4\log n$,
the edge collection produced by \ALG is {\em good}.
Towards this goal, we analyze the properties of the edge (multi)set
$R$ collected by a random walk. 
For our theoretical analysis,
we first ignore the first $\mix$ many steps in $R$.
Abusing notation, we reuse $R$ to denote the remaining edges.
We denote sum of degrees of the edges in $R$ by $d_R=\sum_{e\in R} d_e$.
In a similar vein, we denote by $\nt_R$ the sum of the triangle count
of the edges in $R$: $\nt_R = \sum_{e \in R} \nt_{e}$.
Note that, the stationary distribution over the edges
is the uniform distribution: $\pi(e) = 1/m$. Hence, a fixed edge 
$e$ is part of $R$ with probability $1/m$. 
In the next lemma, we prove that the random variables $d_R$
and $\nt_R$ are tightly concentrated around their mean.

\begin{lemma}[Analysis of $R$]
\label{lem:good_set}
Let $\eps>0$ and $c>6$ be some constants, and $\wL = \frac{\log n }{\eps^2}\cdot \frac{m \mix  \nt_{\emax}}{\NT}$. Let $R$, $d_R$, and $\nt_{R}$ be as defined above. Then,
\begin{enumerate}
    \item $\EX[d_R] = |R| \cdot \frac{d_{E}}{m} $ and $\EX [{\NT}_R] = |R|\cdot  \frac{\NT}{m} $.
    \item With probability at least $1-\frac{\eps}{\log n}$, $d_R \leq \EX \left[ d_R \right] \cdot \frac{\log n}{\eps}$.
    \item With probability at least $1-\frac{1}{c\log n}$, 
    $|\nt_R -\EX[\nt_R]| \leq \eps \EX[\nt_R] $.
\end{enumerate}
\end{lemma}

\begin{proof} We first compute the expected value of $d_R$ and $\nt_R$. 
For each index $i\in [|R|]$ in the set $R$, we define two random variables $Y_{i}^{d}$ and $Y_{i}^{\NT}$:
$Y_{i}^{d} = d_{e_i}$, and $Y_{i}^{\NT} = \nt_{e_i}$, where 
$e_i$ is the $i$-th edge in $R$. 
Then, $d_R = \sum_{i=1}^{\wL} Y_{i}^{d}$ and $t_R = \sum_{i=1}^{\wL} Y_{i}^{\NT}$. We have
\begin{align*}
    \EX \left[Y_{i}^{d} \right] &= \sum_{e \in E} \Pr[e_i=e] \cdot \EX 
    \left[ Y_{i}^{d} | e_i = e\right] = \frac{1}{m} \sum_{e \in E} d_e = \frac{d_{E}}{m} \,.
\end{align*}
By linearity of expectation, $\EX [d_R] = |R| \cdot {d_{E}}/{m} $.
Analogously, using the fact that $\sum_{e\in E}\nt_e = \NT$, we get $\EX [\nt_R] = |R|\cdot  \NT /m $. 

We now turn our focus on the concentration of $d_R$. This is achieved by a simple application of Markov inequality.
\begin{align*}
    \Pr \left[ d_R \geq \EX \left[ d_R \right] \cdot \frac{\log n}{\eps} \right] \leq \frac{\eps}{\log n} \,.
\end{align*}
Hence, the second item in the lemma statement follows.

We now prove the third item. To prove a concentration bound on $\nt_R$, 
we fist bound the variance of $\nt_R$ and then apply Chebyshev
inequality (~\Cref{thm:conc}). Note that not all the edges in $R$
are independent --- however, the edges that
are at least $\mix$ many steps apart in the set $R$ are
independent.
We bound the variance as follows.
\begin{align*}
    &\Var[\nt_R] = \EE[\nt_R^2] - \left(\EE[\nt_R]\right)^2 \\
    &= \sum_{|i-j|> \mix} \EE[Y_{i}^{\NT}]\cdot \EE[Y_{j}^{\NT}] + \sum_{|i-j|\leq  \mix} \EE[Y_{i}^{\NT}\cdot Y_{j}^{\NT}] - \left(\EE[\nt_R]\right)^2  \\
    &\leq  \left(\EE[\nt_R]\right)^2 + \sum_{i\in [|R|]} \mix \nt_{\emax} \EE[Y_{i}^{\NT}]
    -\left(\EE[\nt_R]\right)^2 \\
    &\leq \mix \nt_{\emax} |R| \frac{\NT}{m}\,.
\end{align*}
By Chebyshev's inequality, we can upper bound
$\Pr [ \nt_R \notin (1\pm\eps)\EX[\nt_R]]$ by
\begin{align*}
    \frac{\Var[t_R]}{\eps^2 \EX[\nt_R]^2} &= \frac{1}{\eps^2}\cdot \frac{ \mix \cdot \nt_{\emax} \cdot |R| \cdot \NT }{m}
        \cdot \frac{m^2}{|R|^2 \NT^2} \\
    &= \frac{1}{|R|} \cdot \frac{m \cdot \mix \cdot \nt_{\emax}}{\eps^2 \NT} \leq  \frac{1}{c\log n} \,.
\end{align*}
The last inequality follows because $|R| = r = \frac{\log n }{\eps^2}\cdot \frac{m \mix  \nt_{\emax}}{\NT}$.
\end{proof}

Now, we complete the analysis.
\begin{proof}[Proof of~\Cref{thm:tri_count}]
\Cref{lem:good_set} implies, with probability at least 
$1-1/4\log n$, we have $\nt_R \geq (1-\eps) \cdot |R| \cdot \frac{\NT}{m}$
and $d_R \leq \frac{\log n}{\eps }\cdot |R| \cdot \frac{d_E}{m}$.
Hence, $\frac{\nt_R}{d_R} \geq (1-\eps) \cdot \frac{\eps}{\log n} \cdot \frac{\NT}{d_E}$.
This is the first property (~\cref{eqn:prop_1}) for $R$ to 
to be {\em good}. The second property (~\cref{eqn:prop_2}) is true by the item (3) of ~\Cref{lem:good_set}. Hence, $R$ is {\em good} with
probability at least $1-1/4\log n$. Hence, we remove the
condition on~\Cref{lem:tri}, and derive that with probability at least
$1-1/3\log n$,
$    X \in \left[ (1-4\eps) \NT, (1+4\eps)\NT  \right]$.
Re-scaling the parameter $\eps$ appropriately, the accuracy of 
\ALG follows. The number of queries is bounded by $O(\wL+\sL)$.
The space complexity and the running
time of \ALG are both bounded by $\tO(r+\ell)$.
\footnote{To count the number of the collisions in $R_i$ (~\cref{line:R_i} of~\Cref{alg:edge_estimate}),
we use a dictionary of size $O(|R_i|)$ to maintain the frequency of each element in $R_i$. Hence, the space and time complexity 
of~\Cref{alg:edge_estimate} is bounded by $\tO(|R|)$.}
Note that $\wL$ is $\widetilde{O}(m\mix\nt_{\emax}/NT + \sqrt{m}\mix)$ (\Cref{lem:good_set} and \Cref{thm:edge_estimator})
and
$\sL = c\log^2n \cdot d_E/\eps^3T$ (\Cref{lem:goodY}). \Cref{lem:chiba}
asserts that $d_E = O(m\alpha)$.

\end{proof}

%

\section{Experimental Evaluation}
\label{sec:exp}

In this section, we present extensive empirical evaluations
of \ALG. We implement all algorithms in C++ and 
run our experiments on a Cluster with 128 GB DDR4 DRAM memory capacity
and Intel Xeon E5-2650v4 processor running CentOS 7 operating system.
For evaluation, we use a collection of {\em very large graphs} taken from the Network Repository~\cite{graphrepository2013}. 
Our main focus is on massive graphs --- we mainly consider graphs with
more than 100 million edges in our collection.
\footnote{In presenting the number of edges, we consider the sum of degrees of all the vertices, which is twice the number of 
undirected edges.} 
The details of the graphs are given in~\Cref{table:graphs}. We make all graphs simple by removing
duplicate edges and self-loops.

\begin{table}[!ht]
  \caption{Description of our dataset with the key parameters, \#vertices($n$), \#edges($m$), \#triangles($\NT$), \#sum-edge-degrees($d_E$).
  }
  \label{table:graphs}
  \begin{tabular}{ccccc}
    \toprule
    Graph name  &$n$ & $m$ & $\NT$  & $d_E$\\
    \midrule
    soc-orkut & 3M &213M & 525M & 27B\\
    soc-sinaweibo & 59M & 523M & 213M & 41B\\
    soc-twitter-konect & 42M & 2.4B & 34.8B & 1325B\\
    soc-friendster & 66M & 3.6B & 4.2B & 737B \\
    \bottomrule
  \end{tabular}
\end{table}

\mypar{Key Findings}
Based on our experiments, we report four key findings.
\begin{inparaenum}[\bfseries (1)]
\item \ALG achieves high accuracy for all the datasets with 
minimal parameterization. This is remarkable considering that the
key structural properties of the various graphs are quite
different. In all cases, with less than $0.02m$ queries, \ALG
consistently has a median error of less than 2\% and a maximum
error (over 100 runs) of less than 5\%.
\item The variance of the estimation of \ALG is quite small,
and it converges as we increase the length of the
random walk $\wL$.
\item \ALG consistently outperforms other baseline algorithms
on most of the datasets.
In fact, \ALG 
exhibits remarkable accuracy while observing only a tiny
fraction of the graph --- some of the baseline algorithms
are far from converging at that point.
\item The choice of seed vertex does not affect the accuracy
of \ALG. It exhibits almost identical accuracy irrespective of whether it starts from 
a high degree vertex or a low degree vertex.
\end{inparaenum}

\subsection{Implementation Details}
\label{subsec:exp_implementation}
Our algorithm takes $3$ input parameters:
the length of the random walk $\wL$,
the number of sub-samples $\sL$, and an estimate $\widehat{\mix}$ of the mixing time.
In all experiments, we fix $\widehat{\mix}$ to be $25$,
and set $\sL = 0.05 \wL$. We vary $\wL$ to get runs 
for different sample sizes. 
%
For any particular setting of $\wL$, 
$\sL$, $\widehat{\mix}$, and the seed vertex $\seed$, we repeat \ALG 100 times to determine its accuracy. We measure accuracy in terms of
relative error percentage: $|\NT - \textsc{estimate}|\times 100/\NT$.

For comparison against baselines, we study the 
query complexity of each algorithm.  Since we wish to understand how much of a graph needs to be seen,
we explicitly count two types of queries: {\em random neighbor query} and {\em edge query}.
We stress that all algorithms (\ALG and other baselines) query the degree of every vertex that is seen,
typically for various normalizations. Therefore, degree queries are not useful for distinguishing different
algorithms. We note that the number of queries of \ALG
made is $\wL + 2\sL$.
(For every subsampled edge, \ALG makes a random neighbor
and edge query.) 
In all our results, we present in terms of the number of queries made.
We measure the total number of queries made
as a fraction of sum the degrees: $(\textsc{\#queries} *100 / 2m )\% $. 

\subsection{Evaluation of \ALG}
\label{subsec:eval}

We evaluate \ALG on three parameters: accuracy, convergence, and robustness against the choice of 
initial seed vertex. To demonstrate convergence and robustness to the seed vertex, in the main paper, we choose a subset of the datasets. Results are consistent across all datasets.

\mypar{Accuracy} 
\ALG is remarkably accurate across all the graphs,
even when it queries 2\%-3\% of the entire graph.
In~\Cref{fig:intro-fig}, we plot the median and the max relative error 
of \ALG over 100 runs for each dataset. In all these runs, we set $\sL < 0.03m$,
and the total number of queries is at most $0.03m$.
\ALG has a maximum relative error of only 5\%, and 
the median relative error lies below 2\%. Remarkably, for all the datasets
we use the same parameter settings. 

We present further evidence of the excellent performance of 
\ALG. In~\Cref{fig:apx:variance},
we plot the median relative error percentage of \ALG across the four datasets for a fixed setting of parameters --- we restrict \ALG
to visit only 3\% of the edges. We show the variance of the error in the estimation.
Observe that the variance of the estimation is remarkably small.

Finally, in~\Cref{fig:apx:MRE}, we plot
the median relative error percentage of \ALG while
increasing the length of the random walk, $\wL$. 
The behavior of \ALG is stable 
across all datasets. Observe that, for larger datasets, \ALG 
achieves about almost 2\% accuracy even when it sees only 2\% of the edges.

\begin{figure*}[!ht]
  \centering
  \begin{subfigure}[t]{\columnwidth}
  \centering
  \includegraphics[width=0.66\columnwidth]{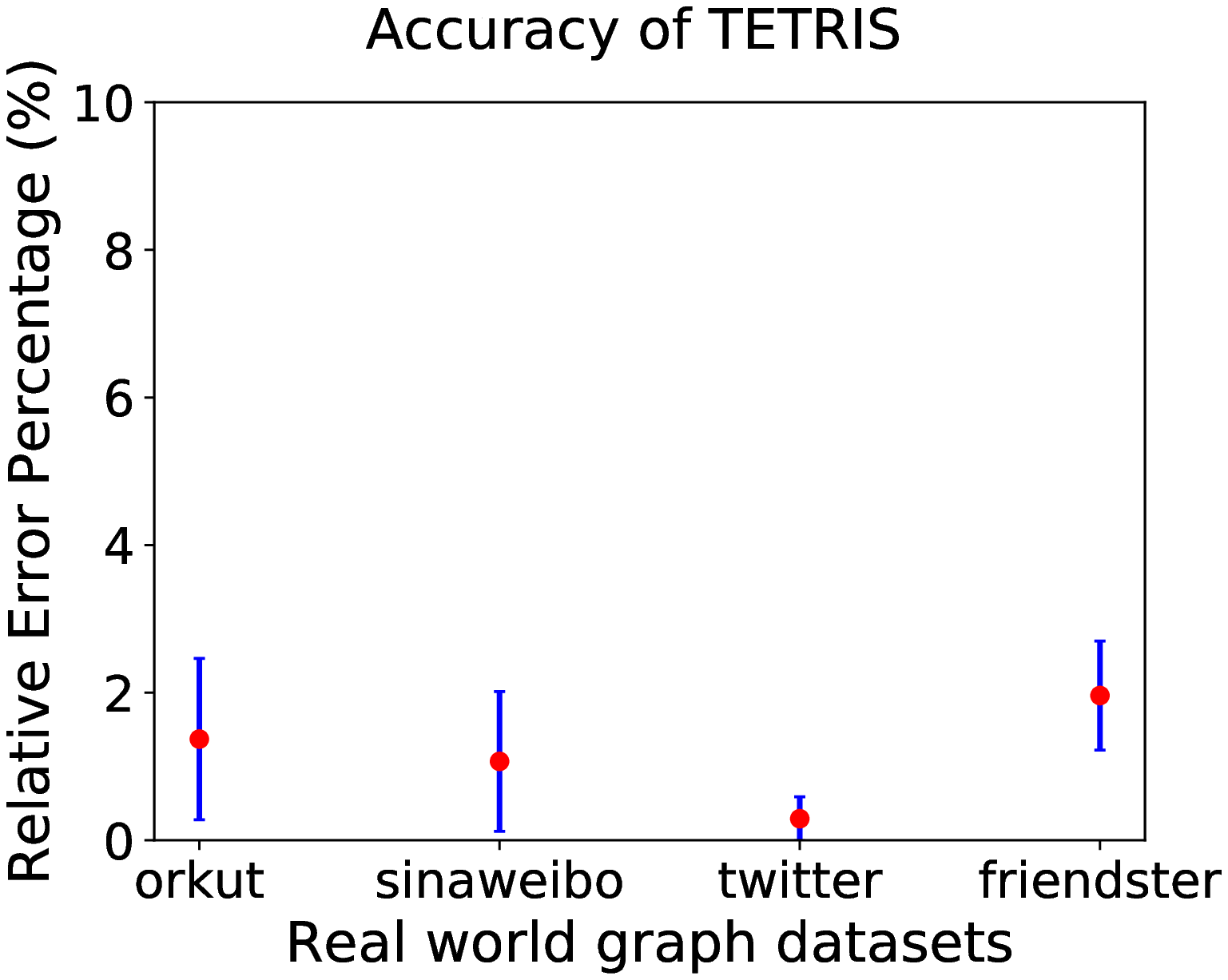} 
  \caption{We plot median relative error
  estimates for \ALG on various datasets for 100 runs with fixed set of parameters. We also show the variance in the error percentage. We restrict \ALG to visit at most 3\% of the edges.}
  \label{fig:apx:variance}
  \end{subfigure} \hfill
  \begin{subfigure}[t]{\columnwidth}
  \centering
  \includegraphics[width=0.66\columnwidth]{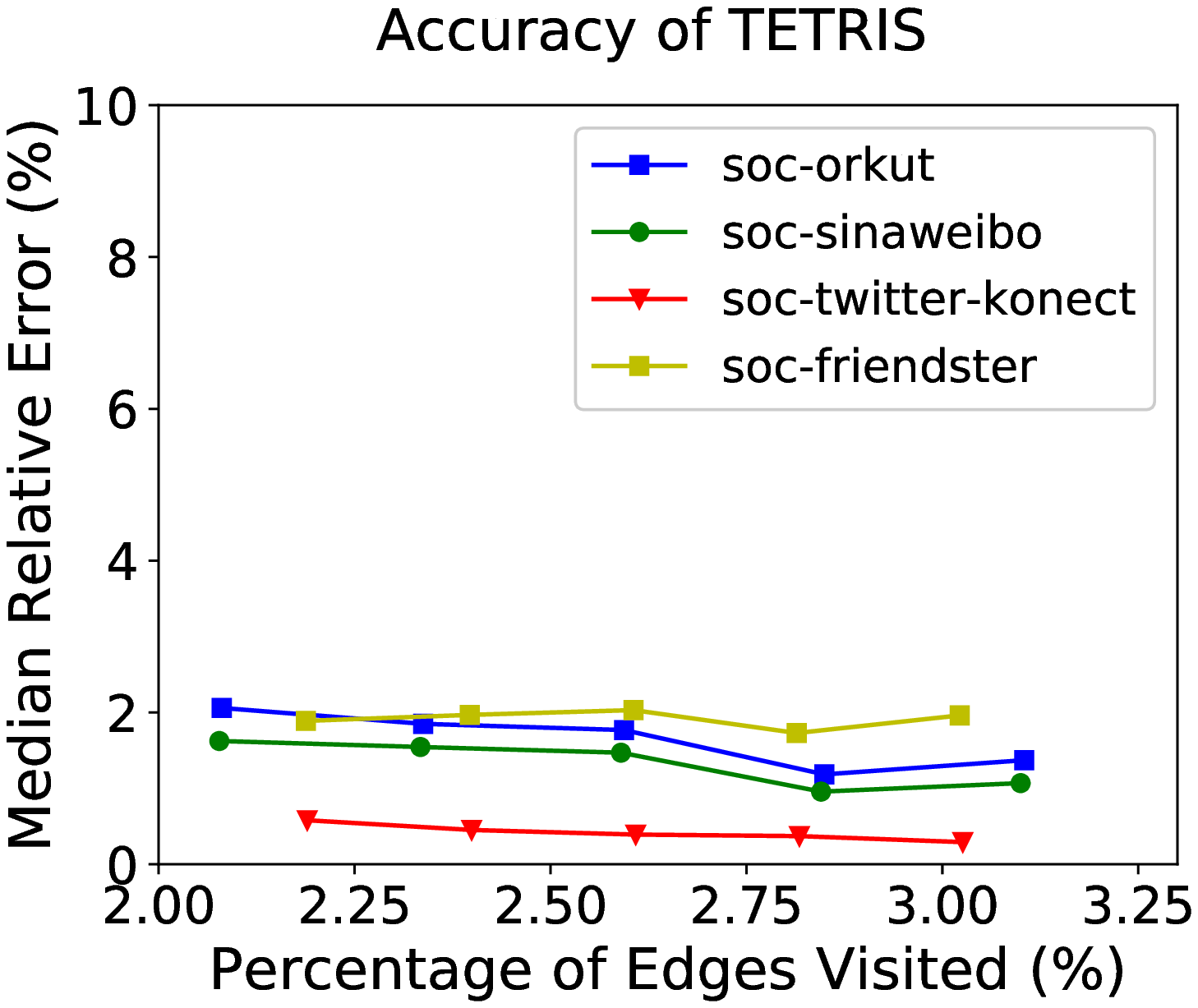}
  \caption{we show the effect of varying the random walk length, $\wL$. On y-axis, we have median relative error percentage, and on x-axis we have percentage of
  the edges visited.}
  \label{fig:apx:MRE}
  \end{subfigure}
  \caption{Accuracy of \ALG. 
  }
  \label{fig:apx:accuracy}
\end{figure*}

\mypar{Convergence}
The convergence for the orkut, weibo and twitter datasets are demonstrated in
~\Cref{fig:converge}. 
We plot the final triangle estimation of \ALG for 100 iterations for fixed choices of $\seed$, $\wL$, $\sL$, and $\mix$ (see~\cref{subsec:exp_implementation} for details). 
We increase $\wL$ gradually and show that the mean estimation tightly concentrates around the true triangle count. 
Observe that the spread of our estimator is less
than 5\% around the true triangle count even when we explore just 3\% of the 
graph. 
%
\begin{figure*}[!ht]
  \centering
  \includegraphics[width=0.66\columnwidth]{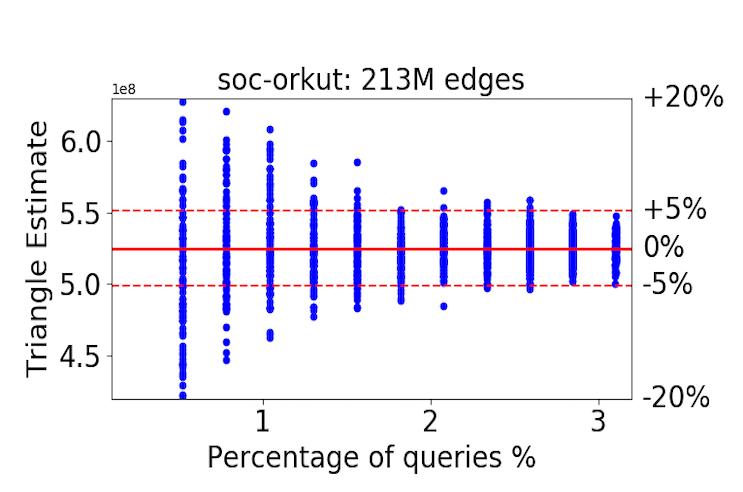}
    \includegraphics[width=0.66\columnwidth]{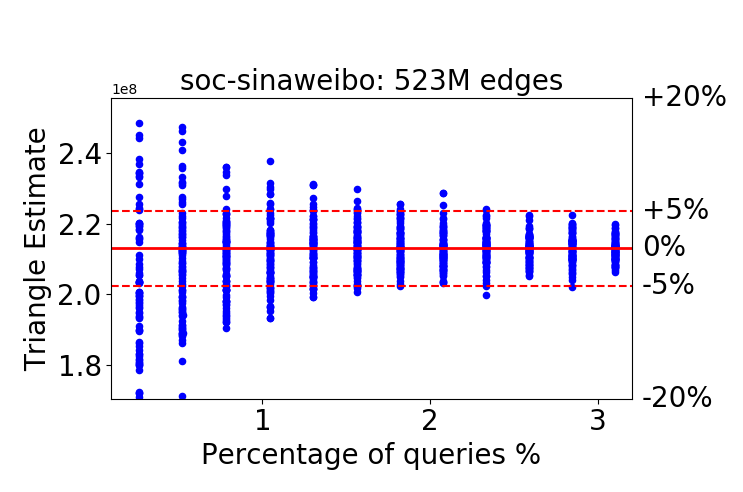}
    \includegraphics[width=0.66\columnwidth]{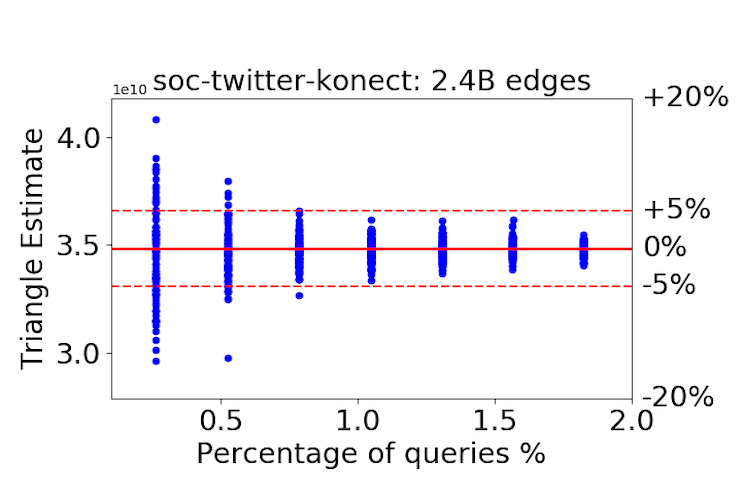}
  \caption{Convergence of \ALG. We plot on the y-axis, the final output of \ALG for each of the 100 runs
  corresponding to fixed value of $\wL$. On x-axis, we show the percentage of the queries
  made by \ALG during its execution by increasing $\wL$. The maximum observed edge percentage corresponding to largest setting of $\wL$ is 3\%.}
  \label{fig:converge}
\end{figure*}

\mypar{Robustness against the choice of seed vertex}
In the previous set of experiments, we selected a
fixed uniform random vertex as $\seed$.
Surprisingly, we show that \ALG performs almost identically regardless of the choice of seed vertex.
We partition the vertices into multiple 
buckets according to their degree: the $i$-th group contains
vertices with degree between $10^{i}$ and $10^{i+1}$. Then, from
each bucket, we select $4$ vertices uniformly at random, and
repeat \ALG $100$ times for a fixed choice of $\wL$. 

In ~\Cref{fig:robustness}, we plot the results for the orkut and the twitter datasets. On the x-axis, we consider
$4$ vertices from each degree based vertex bucket for a total of $16$
vertices. On the y-axis, we plot
the relative median error percentage of \ALG for 100 independent runs 
starting with the corresponding vertex on the x-axis as the seed vertex 
$\seed$. The choice of $\wL$ leads to observing 3\% of the graph.
As we observe, the errors are consistently small irrespective of whether \ALG starts its random walk from a low degree vertex or a high degree vertex. The same behavior persists across datasets for varying $\wL$. 


\begin{figure*}[!ht]
  \centering
  \begin{subfigure}[t]{\columnwidth}
  \centering
     \includegraphics[width=0.9\columnwidth]{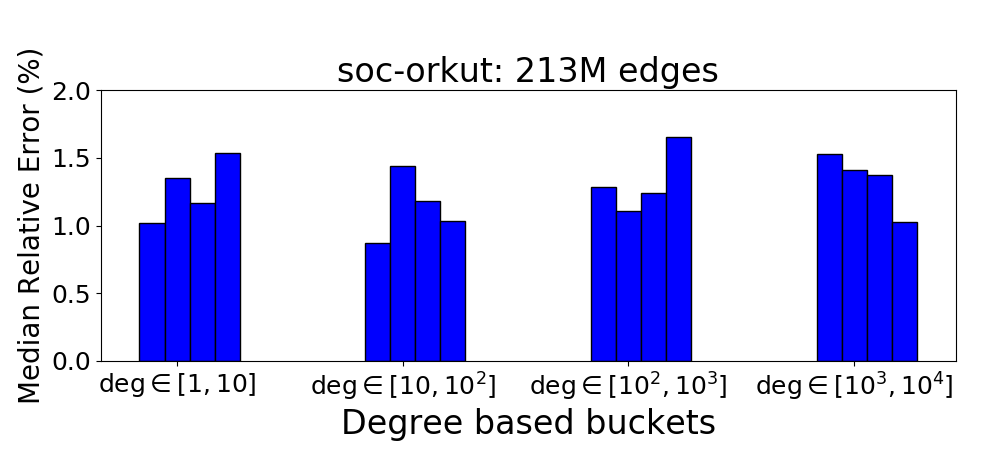}
  \end{subfigure}
   \begin{subfigure}[t]{\columnwidth}
   \centering
     \includegraphics[width=0.9\columnwidth]{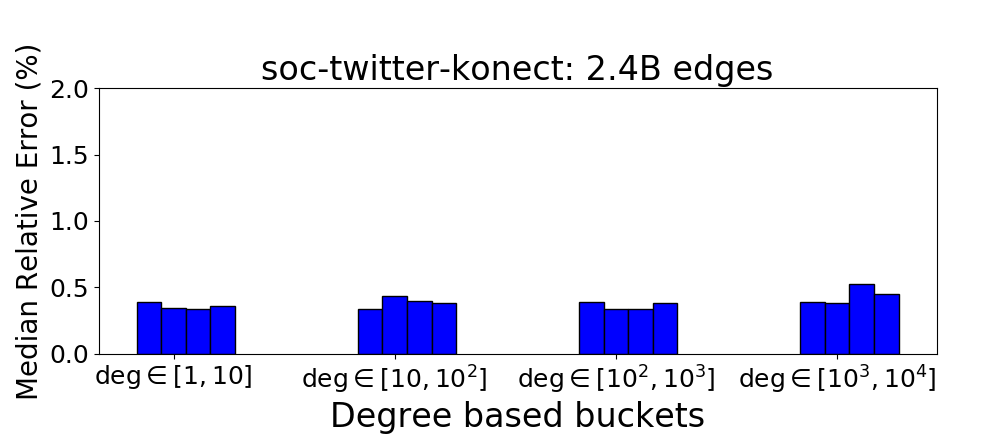} 
     \end{subfigure}
  \caption{Robustness of \ALG. We select 4 vertices uniformly randomly from each
  degree based vertex bucket. On y-axis, we show the median relative error for 100
 runs of \ALG with corresponding vertex as the seed vertex. The  parameters $\wL$
 and $\sL$ are fixed, and results in 3\% of the edges being visited.}
  \label{fig:robustness}
\end{figure*}

\subsection{Comparison against Previous Works}
\label{subsec:comapre_baseline}

\begin{figure*}[!ht]
  \centering
  \includegraphics[width=0.95\textwidth]{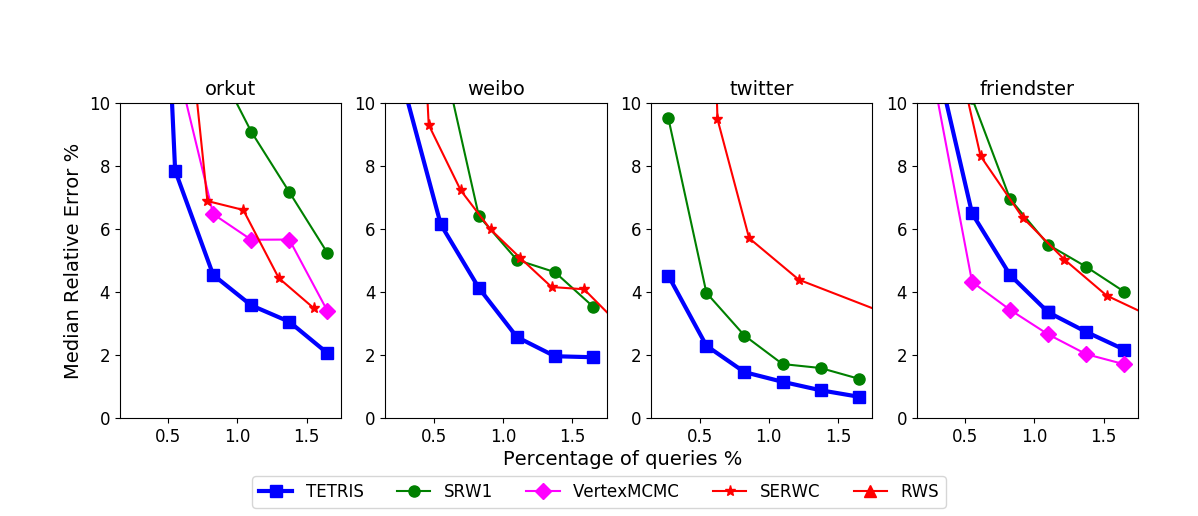}
  \caption{Comparison against baseline. 
  For each dataset and for each parameter setting, we run all the algorithms 100 times
 using the same randomly chosen seed vertex.
 We compare the median relative error in estimation vs 
the percentage of queries made. The median error of RWS does not drop below 10\% for any of the datasets. }
  \label{fig:baseline_comapre}
\end{figure*}

We compare \ALG against the following four benchmark algorithms. 
The first two algorithms are state-of-the-art in the {\em random walk access} model.
The other two algorithms are simulations of widely popular {\em uniform random edge} sampling based algorithms 
in our model (recall that uniform random edge samples are unavailable in our model). 
\begin{enumerate}
    \item \textbf{VertexMCMC (\cite{rahman2014sampling})}: Rahman~\etal in ~\cite{rahman2014sampling} proposed 
    multiple Metropolis-Hastings algorithms, and VertexMCMC is the best amongst them for counting triangles.
    In this procedure, at each step of the random walk, two 
    neighbors of the current vertex are sampled uniformly at random and tested for the existence of a triangle. A final count is derived by scaling by the empirical success probability of finding a triangle. 
    \item \textbf{Subgraph Random Walk (SRW~\cite{chen2016general})}: Chen~\etal in ~\cite{chen2016general}
    proposed SRW1CSS in which every three consecutive vertices
    on the random walk is tested for the existence of a triangle. 
    \footnote{Chen~\etal~\cite{chen2016general} also consider a non-backtracking variation of the random walk in 
    designing various algorithms. However, we do not observe any positive impact on the accuracy due to this variation and
    hence we do not incorporate this.}
    \item \textbf{Random Walk Sparsifier (RWS)}: 
    This algorithm is based on the graph sparsification strategy of ~\cite{tsourakakis2009doulion}.
    It performs a random walk and counts the number of triangles in 
    the multi-graph induced by the edges collected during 
    the walk (with appropriate normalization).
    \item \textbf{Sample Edge by Random Walk and Count (SERWC)}: This algorithm
    is similar in spirit to that of~\cite{wu2016counting}. It counts the 
    number of triangles incident on each
    edge of the random walk and outputs a scaled average as the final estimate.
    We note that SERWC relies on counting the number of triangles
incident on an edge --- a non-trivial task in our model.
To make a fair comparison, we allow it
to make neighbor queries: given $v\in V$ and an integer $i$,
return the $i$-th neighbor of $v$, if available, else return $\emptyset$.
These queries are counted towards final accounting.
\end{enumerate}

We plot our comparison in~\Cref{fig:baseline_comapre} for each of the four datasets. \ALG is consistently accurate over all the datasets. VertexMCMC has better accuracy on the friendster dataset,
however on weibo and twitter it has more than 10\% error even with 1.5\% queries. In contrast,
\ALG converges to an error of less than 2\% with the same amount of queries.
We also observe that RWS
does not converge at all, and its error is more than 10\%. (Essentially, the edges collected by the random walk are too correlated
for standard estimators to work.) We note that SRW is also consistent across all datasets, though
\ALG significantly outperforms it on all datasets.
%

\mypar{Normalization Factor}
Other than VertexMCMC, all the remaining algorithms require
an estimate for the number of edges in the graph. VertexMCMC requires the
wedge count $\sum_{v\in V}\binom{d_v}{2}$.
While such estimates are readily available in various models, the random walk
access model does not reveal this information easily. We use the 
\EdgeALG algorithm with collected edge samples to estimate $m$ for each algorithm that 
requires an estimate for $m$. 
To estimate wedge count, we build a simple unbiased estimator (recall that
the degree of each of the vertices explored by VertexMCMC is available for free).

\begin{acks}

The authors would like to thank the anonymous reviewers for
their valuable feedback. The authors are supported by NSF TRIPODS grant CCF-1740850, NSF CCF-1813165, CCF-1909790, and ARO Award W911NF1910294.

\end{acks}
\bibliographystyle{ACM-Reference-Format}
\bibliography{refs.bib}


\begin{thebibliography}{57}


\ifx \showCODEN    \undefined \def \showCODEN     #1{\unskip}     \fi
\ifx \showDOI      \undefined \def \showDOI       #1{#1}\fi
\ifx \showISBNx    \undefined \def \showISBNx     #1{\unskip}     \fi
\ifx \showISBNxiii \undefined \def \showISBNxiii  #1{\unskip}     \fi
\ifx \showISSN     \undefined \def \showISSN      #1{\unskip}     \fi
\ifx \showLCCN     \undefined \def \showLCCN      #1{\unskip}     \fi
\ifx \shownote     \undefined \def \shownote      #1{#1}          \fi
\ifx \showarticletitle \undefined \def \showarticletitle #1{#1}   \fi
\ifx \showURL      \undefined \def \showURL       {\relax}        \fi
\providecommand\bibfield[2]{#2}
\providecommand\bibinfo[2]{#2}
\providecommand\natexlab[1]{#1}
\providecommand\showeprint[2][]{arXiv:#2}

\bibitem[\protect\citeauthoryear{Ahmed, Duffield, Neville, and Kompella}{Ahmed
  et~al\mbox{.}}{2014a}]%
        {ahmed2014graph}
\bibfield{author}{\bibinfo{person}{Nesreen~K Ahmed}, \bibinfo{person}{Nick
  Duffield}, \bibinfo{person}{Jennifer Neville}, {and} \bibinfo{person}{Ramana
  Kompella}.} \bibinfo{year}{2014}\natexlab{a}.
\newblock \showarticletitle{Graph sample and hold: A framework for big-graph
  analytics}. In \bibinfo{booktitle}{\emph{SIGKDD}}. ACM,
  \bibinfo{publisher}{ACM}, \bibinfo{pages}{1446--1455}.
\newblock


\bibitem[\protect\citeauthoryear{Ahmed, Neville, and Kompella}{Ahmed
  et~al\mbox{.}}{2014b}]%
        {ahmed2014network}
\bibfield{author}{\bibinfo{person}{Nesreen~K Ahmed}, \bibinfo{person}{Jennifer
  Neville}, {and} \bibinfo{person}{Ramana Kompella}.}
  \bibinfo{year}{2014}\natexlab{b}.
\newblock \showarticletitle{Network sampling: From static to streaming graphs}.
\newblock \bibinfo{journal}{\emph{ACM Transactions on Knowledge Discovery from
  Data (TKDD)}} \bibinfo{volume}{8}, \bibinfo{number}{2}
  (\bibinfo{year}{2014}), \bibinfo{pages}{7}.
\newblock


\bibitem[\protect\citeauthoryear{Arifuzzaman, Khan, and Marathe}{Arifuzzaman
  et~al\mbox{.}}{2013}]%
        {arifuzzaman2013patric}
\bibfield{author}{\bibinfo{person}{Shaikh Arifuzzaman}, \bibinfo{person}{Maleq
  Khan}, {and} \bibinfo{person}{Madhav Marathe}.}
  \bibinfo{year}{2013}\natexlab{}.
\newblock \showarticletitle{Patric: A parallel algorithm for counting triangles
  in massive networks}. In \bibinfo{booktitle}{\emph{Conference on Information
  and Knowledge Management (CIKM)}}. ACM, \bibinfo{pages}{529--538}.
\newblock


\bibitem[\protect\citeauthoryear{Azad, Bulu{\c{c}}, and Gilbert}{Azad
  et~al\mbox{.}}{2015}]%
        {azad2015parallel}
\bibfield{author}{\bibinfo{person}{Ariful Azad}, \bibinfo{person}{Aydin
  Bulu{\c{c}}}, {and} \bibinfo{person}{John Gilbert}.}
  \bibinfo{year}{2015}\natexlab{}.
\newblock \showarticletitle{Parallel triangle counting and enumeration using
  matrix algebra}. In \bibinfo{booktitle}{\emph{2015 IEEE International
  Parallel and Distributed Processing Symposium Workshop}}. IEEE,
  \bibinfo{pages}{804--811}.
\newblock


\bibitem[\protect\citeauthoryear{{Bar-Yossef}, Kumar, and
  Sivakumar}{{Bar-Yossef} et~al\mbox{.}}{2002}]%
        {BarYossefKS02}
\bibfield{author}{\bibinfo{person}{Ziv {Bar-Yossef}}, \bibinfo{person}{Ravi
  Kumar}, {and} \bibinfo{person}{D. Sivakumar}.}
  \bibinfo{year}{2002}\natexlab{}.
\newblock \showarticletitle{Reductions in Streaming Algorithms, with an
  Application to Counting Triangles in Graphs}. In
  \bibinfo{booktitle}{\emph{Proc. 13th Symposium on Discrete Algorithms
  (SODA)}}. \bibinfo{pages}{623--632}.
\newblock


\bibitem[\protect\citeauthoryear{Becchetti, Boldi, Castillo, and
  Gionis}{Becchetti et~al\mbox{.}}{2008}]%
        {BecchettiBCG08}
\bibfield{author}{\bibinfo{person}{Luca Becchetti}, \bibinfo{person}{Paolo
  Boldi}, \bibinfo{person}{Carlos Castillo}, {and} \bibinfo{person}{Aristides
  Gionis}.} \bibinfo{year}{2008}\natexlab{}.
\newblock \showarticletitle{Efficient semi-streaming algorithms for local
  triangle counting in massive graphs}. In \bibinfo{booktitle}{\emph{Knowledge
  Data and Discovery (KDD)}}. ACM, \bibinfo{pages}{16--24}.
\newblock


\bibitem[\protect\citeauthoryear{Bera and Chakrabarti}{Bera and
  Chakrabarti}{2017}]%
        {bera2017towards}
\bibfield{author}{\bibinfo{person}{Suman~K Bera} {and} \bibinfo{person}{Amit
  Chakrabarti}.} \bibinfo{year}{2017}\natexlab{}.
\newblock \showarticletitle{Towards tighter space bounds for counting triangles
  and other substructures in graph streams}. In \bibinfo{booktitle}{\emph{Proc.
  34th International Symposium on Theoretical Aspects of Computer Science}}.
\newblock


\bibitem[\protect\citeauthoryear{Bera and Seshadhri}{Bera and
  Seshadhri}{2020}]%
        {BeraDegeneracy}
\bibfield{author}{\bibinfo{person}{Suman~K. Bera} {and} \bibinfo{person}{C.
  Seshadhri}.} \bibinfo{year}{2020}\natexlab{}.
\newblock \showarticletitle{How the Degeneracy Helps for Triangle Counting in
  Graph Streams}. In \bibinfo{booktitle}{\emph{Proceedings of the 39th ACM
  SIGMOD-SIGACT-SIGAI Symposium on Principles of Database Systems}}.
  \bibinfo{pages}{457–467}.
\newblock


\bibitem[\protect\citeauthoryear{Buriol, Frahling, Leonardi,
  Marchetti-Spaccamela, and Sohler}{Buriol et~al\mbox{.}}{2006}]%
        {Buriol2006}
\bibfield{author}{\bibinfo{person}{Luciana~S. Buriol}, \bibinfo{person}{Gereon
  Frahling}, \bibinfo{person}{Stefano Leonardi}, \bibinfo{person}{Alberto
  Marchetti-Spaccamela}, {and} \bibinfo{person}{Christian Sohler}.}
  \bibinfo{year}{2006}\natexlab{}.
\newblock \showarticletitle{Counting Triangles in Data Streams}. In
  \bibinfo{booktitle}{\emph{Proc. 25th ACM Symposium on Principles of Database
  Systems}}. \bibinfo{pages}{253--262}.
\newblock


\bibitem[\protect\citeauthoryear{Chen, Li, Wang, and Lui}{Chen
  et~al\mbox{.}}{2016}]%
        {chen2016general}
\bibfield{author}{\bibinfo{person}{Xiaowei Chen}, \bibinfo{person}{Yongkun Li},
  \bibinfo{person}{Pinghui Wang}, {and} \bibinfo{person}{John~CS Lui}.}
  \bibinfo{year}{2016}\natexlab{}.
\newblock \showarticletitle{A General Framework for Estimating Graphlet
  Statistics via Random Walk}.
\newblock \bibinfo{journal}{\emph{Proceedings of the VLDB Endowment}}
  \bibinfo{volume}{10}, \bibinfo{number}{3} (\bibinfo{year}{2016}).
\newblock


\bibitem[\protect\citeauthoryear{Chiba and Nishizeki}{Chiba and
  Nishizeki}{1985}]%
        {Chiba1985}
\bibfield{author}{\bibinfo{person}{Norishige Chiba} {and}
  \bibinfo{person}{Takao Nishizeki}.} \bibinfo{year}{1985}\natexlab{}.
\newblock \showarticletitle{Arboricity and Subgraph Listing Algorithms}.
\newblock \bibinfo{journal}{\emph{SIAM J. Comput.}} \bibinfo{volume}{14},
  \bibinfo{number}{1} (\bibinfo{year}{1985}), \bibinfo{pages}{210--223}.
\newblock


\bibitem[\protect\citeauthoryear{Chierichetti, Dasgupta, Kumar, Lattanzi, and
  Sarlos}{Chierichetti et~al\mbox{.}}{2016}]%
        {ChDa+16}
\bibfield{author}{\bibinfo{person}{F. Chierichetti}, \bibinfo{person}{A.
  Dasgupta}, \bibinfo{person}{R. Kumar}, \bibinfo{person}{S. Lattanzi}, {and}
  \bibinfo{person}{T. Sarlos}.} \bibinfo{year}{2016}\natexlab{}.
\newblock \showarticletitle{On Sampling Nodes in a Network}. In
  \bibinfo{booktitle}{\emph{Conference on the World Wide Web (WWW)}}.
\newblock


\bibitem[\protect\citeauthoryear{Chierichetti and Haddadan}{Chierichetti and
  Haddadan}{2018}]%
        {chierichetti2018complexity}
\bibfield{author}{\bibinfo{person}{Flavio Chierichetti} {and}
  \bibinfo{person}{Shahrzad Haddadan}.} \bibinfo{year}{2018}\natexlab{}.
\newblock \showarticletitle{On the Complexity of Sampling Vertices Uniformly
  from a Graph}. In \bibinfo{booktitle}{\emph{Proc. 45th International
  Colloquium on Automata, Languages and Programming}}.
\newblock


\bibitem[\protect\citeauthoryear{Cohen}{Cohen}{2009}]%
        {cohen2009graph}
\bibfield{author}{\bibinfo{person}{Jonathan Cohen}.}
  \bibinfo{year}{2009}\natexlab{}.
\newblock \showarticletitle{Graph twiddling in a mapreduce world}.
\newblock \bibinfo{journal}{\emph{Computing in Science \& Engineering}}
  \bibinfo{volume}{11}, \bibinfo{number}{4} (\bibinfo{year}{2009}),
  \bibinfo{pages}{29}.
\newblock


\bibitem[\protect\citeauthoryear{Cooper, Radzik, and Siantos}{Cooper
  et~al\mbox{.}}{2014}]%
        {cooper2014estimating}
\bibfield{author}{\bibinfo{person}{Colin Cooper}, \bibinfo{person}{Tomasz
  Radzik}, {and} \bibinfo{person}{Yiannis Siantos}.}
  \bibinfo{year}{2014}\natexlab{}.
\newblock \showarticletitle{Estimating network parameters using random walks}.
\newblock \bibinfo{journal}{\emph{Social Network Analysis and Mining}}
  \bibinfo{volume}{4}, \bibinfo{number}{1} (\bibinfo{year}{2014}),
  \bibinfo{pages}{168}.
\newblock


\bibitem[\protect\citeauthoryear{Dasgupta, Kumar, and Sarlos}{Dasgupta
  et~al\mbox{.}}{2014}]%
        {DaKu14}
\bibfield{author}{\bibinfo{person}{A. Dasgupta}, \bibinfo{person}{R. Kumar},
  {and} \bibinfo{person}{T. Sarlos}.} \bibinfo{year}{2014}\natexlab{}.
\newblock \showarticletitle{On estimating the average degree}. In
  \bibinfo{booktitle}{\emph{Conference on the World Wide Web (WWW)}}.
  \bibinfo{pages}{795--806}.
\newblock


\bibitem[\protect\citeauthoryear{Eden, Levi, Ron, and Seshadhri}{Eden
  et~al\mbox{.}}{2015}]%
        {ELRS15}
\bibfield{author}{\bibinfo{person}{T. Eden}, \bibinfo{person}{A. Levi},
  \bibinfo{person}{D. Ron}, {and} \bibinfo{person}{C. Seshadhri}.}
  \bibinfo{year}{2015}\natexlab{}.
\newblock \showarticletitle{Approximately Counting Triangles in Sublinear
  Time}. In \bibinfo{booktitle}{\emph{Annual IEEE Symposium on Foundations of
  Computer Science}}, \bibfield{editor}{\bibinfo{person}{GRS11}} (Ed.).
  \bibinfo{pages}{614--633}.
\newblock


\bibitem[\protect\citeauthoryear{Eden, Ron, and Seshadhri}{Eden
  et~al\mbox{.}}{2017}]%
        {ERS17}
\bibfield{author}{\bibinfo{person}{T. Eden}, \bibinfo{person}{D. Ron}, {and}
  \bibinfo{person}{C. Seshadhri}.} \bibinfo{year}{2017}\natexlab{}.
\newblock \showarticletitle{Sublinear Time Estimation of Degree Distribution
  Moments: The Degeneracy Connection}. In
  \bibinfo{booktitle}{\emph{International Colloquium on Automata, Languages and
  Programming}}, \bibfield{editor}{\bibinfo{person}{GRS11}} (Ed.).
  \bibinfo{pages}{614--633}.
\newblock


\bibitem[\protect\citeauthoryear{Eden, Ron, and Seshadhri}{Eden
  et~al\mbox{.}}{2020}]%
        {eden2020faster}
\bibfield{author}{\bibinfo{person}{Talya Eden}, \bibinfo{person}{Dana Ron},
  {and} \bibinfo{person}{C Seshadhri}.} \bibinfo{year}{2020}\natexlab{}.
\newblock \showarticletitle{Faster sublinear approximations of $ k $-cliques
  for low arboricity graphs}. In \bibinfo{booktitle}{\emph{Symposium on
  Discrete Algorithms (SODA)}}.
\newblock


\bibitem[\protect\citeauthoryear{Etemadi, Lu, and Tsin}{Etemadi
  et~al\mbox{.}}{2016}]%
        {etemadi2016efficient}
\bibfield{author}{\bibinfo{person}{Roohollah Etemadi}, \bibinfo{person}{Jianguo
  Lu}, {and} \bibinfo{person}{Yung~H Tsin}.} \bibinfo{year}{2016}\natexlab{}.
\newblock \showarticletitle{Efficient estimation of triangles in very large
  graphs}. In \bibinfo{booktitle}{\emph{Conference on Information and Knowledge
  Management (CIKM)}}. ACM, \bibinfo{pages}{1251--1260}.
\newblock


\bibitem[\protect\citeauthoryear{Green, Yalamanchili, and Mungu{\'\i}a}{Green
  et~al\mbox{.}}{2014}]%
        {green2014fast}
\bibfield{author}{\bibinfo{person}{Oded Green}, \bibinfo{person}{Pavan
  Yalamanchili}, {and} \bibinfo{person}{Llu{\'\i}s-Miquel Mungu{\'\i}a}.}
  \bibinfo{year}{2014}\natexlab{}.
\newblock \showarticletitle{Fast triangle counting on the GPU}. In
  \bibinfo{booktitle}{\emph{Proceedings of the 4th Workshop on Irregular
  Applications: Architectures and Algorithms}}. IEEE Press,
  \bibinfo{pages}{1--8}.
\newblock


\bibitem[\protect\citeauthoryear{Hardiman, Richmond, and Hutzler}{Hardiman
  et~al\mbox{.}}{2009}]%
        {hardiman2009calculating}
\bibfield{author}{\bibinfo{person}{Stephen~James Hardiman},
  \bibinfo{person}{Peter Richmond}, {and} \bibinfo{person}{Stefan Hutzler}.}
  \bibinfo{year}{2009}\natexlab{}.
\newblock \showarticletitle{Calculating statistics of complex networks through
  random walks with an application to the on-line social network Bebo}.
\newblock \bibinfo{journal}{\emph{The European Physical Journal B}}
  \bibinfo{volume}{71}, \bibinfo{number}{4} (\bibinfo{year}{2009}),
  \bibinfo{pages}{611}.
\newblock


\bibitem[\protect\citeauthoryear{Hu, Tao, and Chung}{Hu et~al\mbox{.}}{2014}]%
        {hu2014efficient}
\bibfield{author}{\bibinfo{person}{Xiaocheng Hu}, \bibinfo{person}{Yufei Tao},
  {and} \bibinfo{person}{Chin-Wan Chung}.} \bibinfo{year}{2014}\natexlab{}.
\newblock \showarticletitle{I/O-efficient algorithms on triangle listing and
  counting}.
\newblock \bibinfo{journal}{\emph{ACM Transactions on Database Systems (TODS)}}
  \bibinfo{volume}{39}, \bibinfo{number}{4} (\bibinfo{year}{2014}),
  \bibinfo{pages}{27}.
\newblock


\bibitem[\protect\citeauthoryear{III, Fond, Moreno, and Neville}{III
  et~al\mbox{.}}{2012}]%
        {PfFo+12}
\bibfield{author}{\bibinfo{person}{Joseph J.~Pfeiffer III},
  \bibinfo{person}{Timothy~La Fond}, \bibinfo{person}{Sebasti{\'{a}}n Moreno},
  {and} \bibinfo{person}{Jennifer Neville}.} \bibinfo{year}{2012}\natexlab{}.
\newblock \showarticletitle{Fast Generation of Large Scale Social Networks
  While Incorporating Transitive Closures}. In
  \bibinfo{booktitle}{\emph{International Conference on Privacy, Security, Risk
  and Trust (PASSAT)}}. \bibinfo{pages}{154--165}.
\newblock


\bibitem[\protect\citeauthoryear{III, Moreno, Fond, Neville, and Gallagher}{III
  et~al\mbox{.}}{2014}]%
        {PfMo+14}
\bibfield{author}{\bibinfo{person}{Joseph J.~Pfeiffer III},
  \bibinfo{person}{Sebasti{\'{a}}n Moreno}, \bibinfo{person}{Timothy~La Fond},
  \bibinfo{person}{Jennifer Neville}, {and} \bibinfo{person}{Brian Gallagher}.}
  \bibinfo{year}{2014}\natexlab{}.
\newblock \showarticletitle{Attributed graph models: modeling network structure
  with correlated attributes}. In \bibinfo{booktitle}{\emph{Conference on the
  World Wide Web (WWW)}}. \bibinfo{pages}{831--842}.
\newblock


\bibitem[\protect\citeauthoryear{Itai and Rodeh}{Itai and Rodeh}{1978}]%
        {itai1978finding}
\bibfield{author}{\bibinfo{person}{Alon Itai} {and} \bibinfo{person}{Michael
  Rodeh}.} \bibinfo{year}{1978}\natexlab{}.
\newblock \showarticletitle{Finding a minimum circuit in a graph}.
\newblock \bibinfo{journal}{\emph{SIAM J. Comput.}} \bibinfo{volume}{7},
  \bibinfo{number}{4} (\bibinfo{year}{1978}), \bibinfo{pages}{413--423}.
\newblock


\bibitem[\protect\citeauthoryear{Jha, Pinar, and Seshadhri}{Jha
  et~al\mbox{.}}{2015}]%
        {jha2015counting}
\bibfield{author}{\bibinfo{person}{Madhav Jha}, \bibinfo{person}{Ali Pinar},
  {and} \bibinfo{person}{C Seshadhri}.} \bibinfo{year}{2015}\natexlab{}.
\newblock \showarticletitle{Counting triangles in real-world graph streams:
  Dealing with repeated edges and time windows}. In
  \bibinfo{booktitle}{\emph{2015 49th Asilomar Conference on Signals, Systems
  and Computers}}. IEEE, \bibinfo{pages}{1507--1514}.
\newblock


\bibitem[\protect\citeauthoryear{Jha, Seshadhri, and Pinar}{Jha
  et~al\mbox{.}}{2013}]%
        {jha2013space}
\bibfield{author}{\bibinfo{person}{Madhav Jha}, \bibinfo{person}{C Seshadhri},
  {and} \bibinfo{person}{Ali Pinar}.} \bibinfo{year}{2013}\natexlab{}.
\newblock \showarticletitle{A space efficient streaming algorithm for triangle
  counting using the birthday paradox}. In \bibinfo{booktitle}{\emph{SIGKDD}}.
  ACM, \bibinfo{pages}{589--597}.
\newblock


\bibitem[\protect\citeauthoryear{Jowhari and Ghodsi}{Jowhari and
  Ghodsi}{2005}]%
        {Jowhari2005}
\bibfield{author}{\bibinfo{person}{Hossein Jowhari} {and}
  \bibinfo{person}{Mohammad Ghodsi}.} \bibinfo{year}{2005}\natexlab{}.
\newblock \showarticletitle{New streaming algorithms for counting triangles in
  graphs}.
\newblock In \bibinfo{booktitle}{\emph{Computing and Combinatorics}}.
  \bibinfo{publisher}{Springer}, \bibinfo{pages}{710--716}.
\newblock


\bibitem[\protect\citeauthoryear{Katzir and Hardiman}{Katzir and
  Hardiman}{2015}]%
        {katzir2015estimating}
\bibfield{author}{\bibinfo{person}{Liran Katzir} {and}
  \bibinfo{person}{Stephen~J Hardiman}.} \bibinfo{year}{2015}\natexlab{}.
\newblock \showarticletitle{Estimating clustering coefficients and size of
  social networks via random walk}.
\newblock \bibinfo{journal}{\emph{ACM Transactions on the Web (TWEB)}}
  \bibinfo{volume}{9}, \bibinfo{number}{4} (\bibinfo{year}{2015}),
  \bibinfo{pages}{19}.
\newblock


\bibitem[\protect\citeauthoryear{Katzir, Liberty, and Somekh}{Katzir
  et~al\mbox{.}}{2011}]%
        {katzir2011estimating}
\bibfield{author}{\bibinfo{person}{Liran Katzir}, \bibinfo{person}{Edo
  Liberty}, {and} \bibinfo{person}{Oren Somekh}.}
  \bibinfo{year}{2011}\natexlab{}.
\newblock \showarticletitle{Estimating sizes of social networks via biased
  sampling}. In \bibinfo{booktitle}{\emph{Proceedings of the 20th international
  conference on World wide web}}. ACM, \bibinfo{pages}{597--606}.
\newblock


\bibitem[\protect\citeauthoryear{Khan, Li, Yan, Guan, Chakraborty, and
  Tao}{Khan et~al\mbox{.}}{2011}]%
        {khan2011neighborhood}
\bibfield{author}{\bibinfo{person}{Arijit Khan}, \bibinfo{person}{Nan Li},
  \bibinfo{person}{Xifeng Yan}, \bibinfo{person}{Ziyu Guan},
  \bibinfo{person}{Supriyo Chakraborty}, {and} \bibinfo{person}{Shu Tao}.}
  \bibinfo{year}{2011}\natexlab{}.
\newblock \showarticletitle{Neighborhood based fast graph search in large
  networks}. In \bibinfo{booktitle}{\emph{Proceedings of the 2011 ACM SIGMOD
  International Conference on Management of data}}. \bibinfo{pages}{901--912}.
\newblock


\bibitem[\protect\citeauthoryear{Kim, Lee, Bhowmick, Han, Lee, Ko, and
  Jarrah}{Kim et~al\mbox{.}}{2016}]%
        {kim2016dualsim}
\bibfield{author}{\bibinfo{person}{Hyeonji Kim}, \bibinfo{person}{Juneyoung
  Lee}, \bibinfo{person}{Sourav~S Bhowmick}, \bibinfo{person}{Wook-Shin Han},
  \bibinfo{person}{JeongHoon Lee}, \bibinfo{person}{Seongyun Ko}, {and}
  \bibinfo{person}{Moath~HA Jarrah}.} \bibinfo{year}{2016}\natexlab{}.
\newblock \showarticletitle{DUALSIM: Parallel subgraph enumeration in a massive
  graph on a single machine}. In \bibinfo{booktitle}{\emph{Proceedings of the
  2016 International Conference on Management of Data}}. ACM,
  \bibinfo{pages}{1231--1245}.
\newblock


\bibitem[\protect\citeauthoryear{Kolda, Pinar, Plantenga, Seshadhri, and
  Task}{Kolda et~al\mbox{.}}{2014}]%
        {kolda2014counting}
\bibfield{author}{\bibinfo{person}{Tamara~G Kolda}, \bibinfo{person}{Ali
  Pinar}, \bibinfo{person}{Todd Plantenga}, \bibinfo{person}{C Seshadhri},
  {and} \bibinfo{person}{Christine Task}.} \bibinfo{year}{2014}\natexlab{}.
\newblock \showarticletitle{Counting triangles in massive graphs with
  MapReduce}.
\newblock \bibinfo{journal}{\emph{SIAM Journal on Scientific Computing}}
  \bibinfo{volume}{36}, \bibinfo{number}{5} (\bibinfo{year}{2014}),
  \bibinfo{pages}{S48--S77}.
\newblock


\bibitem[\protect\citeauthoryear{Kolountzakis, Miller, Peng, and
  Tsourakakis}{Kolountzakis et~al\mbox{.}}{2012}]%
        {kolountzakis2012efficient}
\bibfield{author}{\bibinfo{person}{Mihail~N Kolountzakis},
  \bibinfo{person}{Gary~L Miller}, \bibinfo{person}{Richard Peng}, {and}
  \bibinfo{person}{Charalampos~E Tsourakakis}.}
  \bibinfo{year}{2012}\natexlab{}.
\newblock \showarticletitle{Efficient triangle counting in large graphs via
  degree-based vertex partitioning}.
\newblock \bibinfo{journal}{\emph{Internet Mathematics}} \bibinfo{volume}{8},
  \bibinfo{number}{1-2} (\bibinfo{year}{2012}), \bibinfo{pages}{161--185}.
\newblock


\bibitem[\protect\citeauthoryear{Leskovec and Faloutsos}{Leskovec and
  Faloutsos}{2006}]%
        {LF06}
\bibfield{author}{\bibinfo{person}{Jure Leskovec} {and}
  \bibinfo{person}{Christos Faloutsos}.} \bibinfo{year}{2006}\natexlab{}.
\newblock \showarticletitle{Sampling from large graphs}. In
  \bibinfo{booktitle}{\emph{Knowledge Data and Discovery (KDD)}}. ACM,
  \bibinfo{pages}{631--636}.
\newblock


\bibitem[\protect\citeauthoryear{Maiya and Berger-Wolf}{Maiya and
  Berger-Wolf}{2011}]%
        {maiya2011benefits}
\bibfield{author}{\bibinfo{person}{Arun~S Maiya} {and} \bibinfo{person}{Tanya~Y
  Berger-Wolf}.} \bibinfo{year}{2011}\natexlab{}.
\newblock \showarticletitle{Benefits of bias: Towards better characterization
  of network sampling}. In \bibinfo{booktitle}{\emph{Proceedings of the 17th
  ACM SIGKDD international conference on Knowledge discovery and data mining}}.
  ACM, \bibinfo{pages}{105--113}.
\newblock


\bibitem[\protect\citeauthoryear{McGregor, Vorotnikova, and Vu}{McGregor
  et~al\mbox{.}}{2016}]%
        {McGregor2016}
\bibfield{author}{\bibinfo{person}{Andrew McGregor}, \bibinfo{person}{Sofya
  Vorotnikova}, {and} \bibinfo{person}{Hoa~T. Vu}.}
  \bibinfo{year}{2016}\natexlab{}.
\newblock \showarticletitle{Better Algorithms for Counting Triangles in Data
  Streams}. In \bibinfo{booktitle}{\emph{ACM Symposium on Principles of
  Database Systems}}. \bibinfo{pages}{401--411}.
\newblock


\bibitem[\protect\citeauthoryear{Milo, Shen-Orr, Itzkovitz, Kashtan,
  Chklovskii, and Alon}{Milo et~al\mbox{.}}{2002}]%
        {milo2002network}
\bibfield{author}{\bibinfo{person}{Ron Milo}, \bibinfo{person}{Shai Shen-Orr},
  \bibinfo{person}{Shalev Itzkovitz}, \bibinfo{person}{Nadav Kashtan},
  \bibinfo{person}{Dmitri Chklovskii}, {and} \bibinfo{person}{Uri Alon}.}
  \bibinfo{year}{2002}\natexlab{}.
\newblock \showarticletitle{Network motifs: simple building blocks of complex
  networks}.
\newblock \bibinfo{journal}{\emph{Science}} \bibinfo{volume}{298},
  \bibinfo{number}{5594} (\bibinfo{year}{2002}), \bibinfo{pages}{824--827}.
\newblock


\bibitem[\protect\citeauthoryear{Newman}{Newman}{2003}]%
        {Ne03}
\bibfield{author}{\bibinfo{person}{M.~E.~J. Newman}.}
  \bibinfo{year}{2003}\natexlab{}.
\newblock \showarticletitle{The Structure and Function of Complex Networks}.
\newblock \bibinfo{journal}{\emph{SIAM Rev.}} \bibinfo{volume}{45},
  \bibinfo{number}{2} (\bibinfo{year}{2003}), \bibinfo{pages}{167--256}.
\newblock
\urldef\tempurl%
\url{https://doi.org/10.1137/S003614450342480}
\showDOI{\tempurl}


\bibitem[\protect\citeauthoryear{Pavan, Tangwongsan, Tirthapura, and Wu}{Pavan
  et~al\mbox{.}}{2013}]%
        {PavanTTW13}
\bibfield{author}{\bibinfo{person}{A. Pavan}, \bibinfo{person}{Kanat
  Tangwongsan}, \bibinfo{person}{Srikanta Tirthapura}, {and}
  \bibinfo{person}{Kun{-}Lung Wu}.} \bibinfo{year}{2013}\natexlab{}.
\newblock \showarticletitle{Counting and Sampling Triangles from a Graph
  Stream}.
\newblock \bibinfo{journal}{\emph{{PVLDB}}} \bibinfo{volume}{6},
  \bibinfo{number}{14} (\bibinfo{year}{2013}), \bibinfo{pages}{1870--1881}.
\newblock


\bibitem[\protect\citeauthoryear{Rahman and Hasan}{Rahman and Hasan}{2014}]%
        {rahman2014sampling}
\bibfield{author}{\bibinfo{person}{Mahmudur Rahman} {and}
  \bibinfo{person}{Mohammad~Al Hasan}.} \bibinfo{year}{2014}\natexlab{}.
\newblock \showarticletitle{Sampling triples from restricted networks using
  MCMC strategy}. In \bibinfo{booktitle}{\emph{Proceedings of the 23rd ACM
  International Conference on Information and Knowledge Management}}.
  \bibinfo{pages}{1519--1528}.
\newblock


\bibitem[\protect\citeauthoryear{Ron and Tsur}{Ron and Tsur}{2016}]%
        {RT16}
\bibfield{author}{\bibinfo{person}{Dana Ron} {and} \bibinfo{person}{Gilad
  Tsur}.} \bibinfo{year}{2016}\natexlab{}.
\newblock \showarticletitle{The Power of an Example: Hidden Set Size
  Approximation Using Group Queries and Conditional Sampling}.
\newblock \bibinfo{journal}{\emph{ACM Transactions on Computation Theory}}
  \bibinfo{volume}{8}, \bibinfo{number}{4} (\bibinfo{year}{2016}),
  \bibinfo{pages}{15:1--15:19}.
\newblock


\bibitem[\protect\citeauthoryear{Rossi and Ahmed}{Rossi and Ahmed}{2013}]%
        {graphrepository2013}
\bibfield{author}{\bibinfo{person}{Ryan Rossi} {and} \bibinfo{person}{Nesreen
  Ahmed}.} \bibinfo{year}{2013}\natexlab{}.
\newblock \bibinfo{title}{Network Repository}.
\newblock
\newblock
\urldef\tempurl%
\url{http://networkrepository.com}
\showURL{%
\tempurl}


\bibitem[\protect\citeauthoryear{Schank and Wagner}{Schank and Wagner}{2005}]%
        {schank2005finding}
\bibfield{author}{\bibinfo{person}{Thomas Schank} {and}
  \bibinfo{person}{Dorothea Wagner}.} \bibinfo{year}{2005}\natexlab{}.
\newblock \showarticletitle{Finding, counting and listing all triangles in
  large graphs, an experimental study}. In
  \bibinfo{booktitle}{\emph{International workshop on experimental and
  efficient algorithms}}. Springer, \bibinfo{pages}{606--609}.
\newblock


\bibitem[\protect\citeauthoryear{Seshadhri, Kolda, and Pinar}{Seshadhri
  et~al\mbox{.}}{2012}]%
        {SeKoPi11}
\bibfield{author}{\bibinfo{person}{C. Seshadhri}, \bibinfo{person}{Tamara~G.
  Kolda}, {and} \bibinfo{person}{Ali Pinar}.} \bibinfo{year}{2012}\natexlab{}.
\newblock \showarticletitle{Community structure and scale-free collections of
  {Erd\"os-R\'enyi} graphs}.
\newblock \bibinfo{journal}{\emph{Physical Review E}} \bibinfo{volume}{85},
  \bibinfo{number}{5} (\bibinfo{date}{May} \bibinfo{year}{2012}),
  \bibinfo{pages}{056109}.
\newblock
\urldef\tempurl%
\url{https://doi.org/10.1103/PhysRevE.85.056109}
\showDOI{\tempurl}


\bibitem[\protect\citeauthoryear{Seshadhri, Pinar, and Kolda}{Seshadhri
  et~al\mbox{.}}{2014}]%
        {seshadhri2014wedge}
\bibfield{author}{\bibinfo{person}{C. Seshadhri}, \bibinfo{person}{Ali Pinar},
  {and} \bibinfo{person}{Tamara~G Kolda}.} \bibinfo{year}{2014}\natexlab{}.
\newblock \showarticletitle{Wedge sampling for computing clustering
  coefficients and triangle counts on large graphs}.
\newblock \bibinfo{journal}{\emph{Statistical Analysis and Data Mining}}
  \bibinfo{volume}{7}, \bibinfo{number}{4} (\bibinfo{year}{2014}),
  \bibinfo{pages}{294--307}.
\newblock


\bibitem[\protect\citeauthoryear{Seshadhri and Tirthapura}{Seshadhri and
  Tirthapura}{2019}]%
        {SeTi19}
\bibfield{author}{\bibinfo{person}{C. Seshadhri} {and}
  \bibinfo{person}{Srikanta Tirthapura}.} \bibinfo{year}{2019}\natexlab{}.
\newblock \showarticletitle{Scalable Subgraph Counting: The Methods Behind The
  Madness: {WWW} 2019 Tutorial}. In \bibinfo{booktitle}{\emph{Conference on the
  World Wide Web (WWW)}}.
\newblock


\bibitem[\protect\citeauthoryear{Stefani, Epasto, Riondato, and Upfal}{Stefani
  et~al\mbox{.}}{2017}]%
        {stefani2017triest}
\bibfield{author}{\bibinfo{person}{Lorenzo~De Stefani},
  \bibinfo{person}{Alessandro Epasto}, \bibinfo{person}{Matteo Riondato}, {and}
  \bibinfo{person}{Eli Upfal}.} \bibinfo{year}{2017}\natexlab{}.
\newblock \showarticletitle{Triest: Counting local and global triangles in
  fully dynamic streams with fixed memory size}.
\newblock \bibinfo{journal}{\emph{ACM Transactions on Knowledge Discovery from
  Data (TKDD)}} \bibinfo{volume}{11}, \bibinfo{number}{4}
  (\bibinfo{year}{2017}), \bibinfo{pages}{43}.
\newblock


\bibitem[\protect\citeauthoryear{Suri and Vassilvitskii}{Suri and
  Vassilvitskii}{2011}]%
        {Suri2011}
\bibfield{author}{\bibinfo{person}{Siddharth Suri} {and}
  \bibinfo{person}{Sergei Vassilvitskii}.} \bibinfo{year}{2011}\natexlab{}.
\newblock \showarticletitle{Counting triangles and the curse of the last
  reducer}. In \bibinfo{booktitle}{\emph{Conference on the World Wide Web
  (WWW)}}. \bibinfo{pages}{607--614}.
\newblock


\bibitem[\protect\citeauthoryear{Tangwongsan, Pavan, and
  Tirthapura}{Tangwongsan et~al\mbox{.}}{2013}]%
        {tangwongsan2013parallel}
\bibfield{author}{\bibinfo{person}{Kanat Tangwongsan}, \bibinfo{person}{Aduri
  Pavan}, {and} \bibinfo{person}{Srikanta Tirthapura}.}
  \bibinfo{year}{2013}\natexlab{}.
\newblock \showarticletitle{Parallel triangle counting in massive streaming
  graphs}. In \bibinfo{booktitle}{\emph{Knowledge Data and Discovery (KDD)}}.
  ACM, \bibinfo{pages}{781--786}.
\newblock


\bibitem[\protect\citeauthoryear{Tsourakakis, Kang, Miller, and
  Faloutsos}{Tsourakakis et~al\mbox{.}}{2009}]%
        {tsourakakis2009doulion}
\bibfield{author}{\bibinfo{person}{Charalampos~E Tsourakakis},
  \bibinfo{person}{U Kang}, \bibinfo{person}{Gary~L Miller}, {and}
  \bibinfo{person}{Christos Faloutsos}.} \bibinfo{year}{2009}\natexlab{}.
\newblock \showarticletitle{Doulion: counting triangles in massive graphs with
  a coin}. In \bibinfo{booktitle}{\emph{Knowledge Data and Discovery (KDD)}}.
  ACM, \bibinfo{pages}{837--846}.
\newblock


\bibitem[\protect\citeauthoryear{Turk and Turkoglu}{Turk and Turkoglu}{2019}]%
        {Turk2019}
\bibfield{author}{\bibinfo{person}{Ata Turk} {and} \bibinfo{person}{Duru
  Turkoglu}.} \bibinfo{year}{2019}\natexlab{}.
\newblock \showarticletitle{Revisiting Wedge Sampling for Triangle Counting}.
  In \bibinfo{booktitle}{\emph{Conference on the World Wide Web (WWW)}}.
  \bibinfo{pages}{1875--1885}.
\newblock
\showISBNx{978-1-4503-6674-8}


\bibitem[\protect\citeauthoryear{T{\"u}rkoglu and Turk}{T{\"u}rkoglu and
  Turk}{2017}]%
        {turkoglu2017edge}
\bibfield{author}{\bibinfo{person}{Duru T{\"u}rkoglu} {and}
  \bibinfo{person}{Ata Turk}.} \bibinfo{year}{2017}\natexlab{}.
\newblock \showarticletitle{Edge-based wedge sampling to estimate triangle
  counts in very large graphs}. In \bibinfo{booktitle}{\emph{International
  Conference on Data Mining (ICDM)}}. \bibinfo{pages}{455--464}.
\newblock


\bibitem[\protect\citeauthoryear{Wasserman and Faust}{Wasserman and
  Faust}{1994}]%
        {wasserman1994}
\bibfield{author}{\bibinfo{person}{Stanley Wasserman} {and}
  \bibinfo{person}{Katherine Faust}.} \bibinfo{year}{1994}\natexlab{}.
\newblock \bibinfo{booktitle}{\emph{Social network analysis: Methods and
  applications}}. Vol.~\bibinfo{volume}{8}.
\newblock \bibinfo{publisher}{Cambridge university press}.
\newblock


\bibitem[\protect\citeauthoryear{Watts and Strogatz}{Watts and
  Strogatz}{1998}]%
        {watts1998collective}
\bibfield{author}{\bibinfo{person}{Duncan~J Watts} {and}
  \bibinfo{person}{Steven~H Strogatz}.} \bibinfo{year}{1998}\natexlab{}.
\newblock \showarticletitle{Collective dynamics of ‘small-world’networks}.
\newblock \bibinfo{journal}{\emph{nature}} \bibinfo{volume}{393},
  \bibinfo{number}{6684} (\bibinfo{year}{1998}), \bibinfo{pages}{440}.
\newblock


\bibitem[\protect\citeauthoryear{Wu, Yi, and Li}{Wu et~al\mbox{.}}{2016}]%
        {wu2016counting}
\bibfield{author}{\bibinfo{person}{Bin Wu}, \bibinfo{person}{Ke Yi}, {and}
  \bibinfo{person}{Zhenguo Li}.} \bibinfo{year}{2016}\natexlab{}.
\newblock \showarticletitle{Counting triangles in large graphs by random
  sampling}.
\newblock \bibinfo{journal}{\emph{IEEE Transactions on Knowledge and Data
  Engineering}} \bibinfo{volume}{28}, \bibinfo{number}{8}
  (\bibinfo{year}{2016}), \bibinfo{pages}{2013--2026}.
\newblock


\end{thebibliography}


\appendix
\section{Proof of Theorem~\ref{thm:edge_estimator}}
\label{appendix:proof}

The proof of the~\Cref{thm:edge_estimator} follows directly from the Theorem 3.1 of~\cite{RT16}. For the sake of completeness, we include the proof here.
\begin{proof}[Proof of~\Cref{thm:edge_estimator}]
We begin with the estimation of $\EX[c_i]$ and $\Var[c_i]$.
Assume for each $1\leq j < k \leq |R_i|$, $c_{j,k}$ be the 
indicator random variable that is set to $1$ if the $j$-th and $k$-th
element in $R_i$ are same. Then,
\begin{align*}
    \EX[c_i] = \sum_{(j,k): i \neq j} \EX[c_{j,k}] 
    = \frac{1}{m} \cdot \binom{|R_i|}{2} \,.
\end{align*}
Now we turn to variance estimation.
Note that each edge in the set $R_i$ is $\widehat{\mix} \geq \mix$ many
steps apart in the set $R$, and hence are independent.
We have, $\Var[c_i] = \EX[c_i^2] - (\EX[c_i])^2$. Now, 
$\EX[c_i^2] = \sum_{i<j,k<l} \EX[c_{ij}c_{kl}]$.
Expanding the summations, we get three types of random variables: (1) $c_{ij}^2$ (2) $c_{ij}\cdot c_{kl}$ with exactly three 
distinct indices among $i,j,k,l$, and (3) $c_{ij}\cdot c_{kl}$ with all
four distinct indices $i,j,k,l$. For the first type, $\EX[c_{ij}^2] = 1/m$. For the remaining two type, $\EX[c_{ij}\cdot c_{kl}] = 1/m^2$.
Plugging in these terms, we get
\begin{align*}
    \Var[c_i] &= O\left( \frac{|R_i|^2}{m} + \frac{|R_i|^3}{m^2} \right)\,.
\end{align*}
Then, we apply the Chebyshev Inequality in~\Cref{thm:conc} and get
\begin{align*}
    \Pr [c_i \notin (1\pm \eps ) \EX[c_i]] \leq O\left( \frac{1}{\eps^2} \left(\frac{m}{|R_i|^2} + \frac{1}{|R_i|}\right)\right) \,.
\end{align*}
Plugging the the value of $|R_i|$, we get $Y_i\in (1\pm \eps)m$ with 
constant probability. We boost the success probability with 
repetitions.
\end{proof}

\end{document}